\documentclass[journal]{IEEEtran}
\ifCLASSINFOpdf
\else
\fi
\usepackage{simplemargins}
\usepackage{graphicx}
\usepackage{indentfirst}
\usepackage{epsfig}
\usepackage{subfigure}
\usepackage{amsfonts}
\usepackage{color}
\usepackage{amsmath}
\usepackage{nomencl}
\usepackage{epstopdf}

\usepackage{amssymb}
\usepackage{multicol}
\usepackage{rotating}
\usepackage{url}
\usepackage{bm}
\usepackage[ruled,commentsnumbered]{algorithm2e}
\usepackage{amsthm}
\usepackage{algpseudocode}
\usepackage[tikz]{bclogo}
\usepackage{booktabs}
\newcommand{\footnotesc}[1]{\footnote{\scriptsize #1}}

\newcounter{MYtempeqncnt}

\newcommand{\be}{\begin{itemize}} \newcommand{\ee}{\end{itemize}}

\newtheorem{lemma}{Lemma}
\newtheorem{definition}{Definition}
\newcommand{\rev}[1]{{\color{black}#1}} 
\makeatletter
\newcommand\figcaption{\def\@captype{figure}\caption}
\newcommand\tabcaption{\def\@captype{table}\caption}
\makeatother
\newcommand{\mytab}{
	\scriptsize
	\begin{tabular}[b]{p{0.19cm}@{\hskip 0.15cm}p{0.31cm}@{\hskip 0.116cm}p{0.31cm}@{\hskip 0.116cm}p{0.31cm}@{\hskip 0.116cm}p{0.31cm}@{\hskip 0.116cm}p{0.31cm}@{\hskip 0.116cm}p{0.31cm}}
		\toprule
		   & \textbf{1} & \textbf{2} & \textbf{3} & \textbf{4} & \textbf{5} & \textbf{6} \\
		\midrule
		\scriptsize
		\textbf{1} & 3.4 & 4.7 & .6 & 0 & 0 & .99 \\
		\textbf{2} & 0 & 0 & 0 & 0 & 0 & 0 \\
		\textbf{3} & 0 & 0 & 1 & 0 & 0 & 0 \\
		\textbf{4} & 0 & 0 & 0 & 1 & 0 & 0 \\
		\textbf{5} & 0 & 0 & .18 & 1.2 & 2.7 & 0 \\
		\textbf{6} & 0 & 0 & .46 & 0 & 0 & 1.6 \\
		\bottomrule
	\end{tabular}
}

\def\eq{\triangleq}

\def\dii{\mbox{...}}

\def\I{\mathcal{I}}

\begin{document}
%
\title{Efficient and Fair Collaborative Mobile Internet Access}
%
%
%

\def\M{\mathcal{M}}
\def\I{\mathcal{I}}
\def\d{\mathrm{d}}

\def\eq{\triangleq}

\author{George~Iosifidis, Lin~Gao, Jianwei~Huang, and~Leandros~Tassiulas
\thanks{George Iosifidis is with the School of Computer Science and Statistics, Trinity College Dublin, and the telecommunications research center CONNECT, Ireland. Lin Gao is with the College of Electronic and Information Engineering, Harbin Institute of Technology, Shenzhen, China. Jianwei Huang is with the Department of Information Engineering, The Chinese University of Hong Kong, Hong Kong. Leandros Tassiulas is with Yale University, Electrical Engineering Department and Institute for Network Science, USA. Parts of this work appeared in Proc. of IEEE Infocom 2014~\cite{opengarden-infocom}.}%
}

%
%

\markboth{IEEE/ACM Transactions on Networking}%
{Iosifidis \MakeLowercase{\textit{et al.}}}
%


\maketitle

\begin{abstract}
The surging global mobile data traffic challenges the economic viability of cellular networks and calls for innovative solutions to reduce the network congestion and improve user experience. In this context, user-provided networks (UPNs), where mobile users share their Internet access by exploiting their diverse network resources and needs, turn out to be very promising. Heterogeneous users with advanced handheld devices can form connections in a distributed fashion and unleash dormant network resources at the network edge. However, the success of such services heavily depends on users' willingness to contribute their resources, such as network access and device battery energy. In this paper, we introduce a general framework for UPN services and design a bargaining-based distributed incentive mechanism to ensure users' participation. The proposed mechanism determines the resources that each user should contribute in order to maximize the aggregate data rate in UPN, and fairly allocate the benefit among the users. The numerical results verify that the service can always improve users' performance, and such improvement increases with the diversity of the users' resources. Quantitatively, it can reach an average $30\%$ increase of the total served traffic for a typical scenario even with only $6$ mobile users.


\end{abstract}

\begin{IEEEkeywords}
Network Economics, Network Optimization, Nash Bargaining, \rev{Fog Computing, User-Provided Networks}
\end{IEEEkeywords}

%
\IEEEpeerreviewmaketitle

\section{Introduction}

\subsection{Background and Motivations}

According to several recent industry reports \cite{cisco-2012}, \cite{ericsson}, mobile data is expected to increase with an annual growth rate of $60\%$ in the next several years, reaching 25 exabytes per month in 2020. This surging traffic places an unprecedented strain on cellular networks, which need to substantially expand their capacities. However, it is clear that the traditional capacity increase strategies of mobile network operators (MNOs), such as acquiring more spectrum or deploying additional network infrastructure, are often time-consuming, costly, and eventually inadequate to accommodate the traffic growth. Therefore, MNOs often end up offering services of low quality \cite{NYTimes-anger}, or charging their subscribers very expensive usage-based fees \cite{itu-report}. This means that a large number of mobile users do not have access to the low-cost and high-speed mobile Internet, and hence there are significant user dissatisfactions and frequent user churns. This leads to the growing consensus that more disruptive methods and forward-looking solutions are needed to resolve the growing gap between data supply and demand.

At the same time, recent technological advancements have resulted in sophisticated user handheld equipments, such as smartphones with Wi-Fi (802.11) and Bluetooth (802.15.1) interfaces, 4G chipsets supporting cellular connections up to 150Mbps \cite{qualcom}, and high-end processors that can execute complicated networking tasks. However, the conventional approach of using these devices as simple transceivers, which are completely controlled by the cellular base stations to serve only the needs of their owners, does not fully exploit their communication and computational capabilities. Clearly, these devices can  also offer communication services to nearby users, by acting as mobile Wi-Fi hotspots or relays. In other words, it is to transform users to local \emph{micro-operators}, serving each other's needs. This leads to the so-called user-provided networks (UPNs) \cite{Sofia-UPC}, \cite{gios-com-mag}, \cite{gios-hybrid}, \cite{gios-wiopt} which constitute a promising solution for alleviating network congestion, reducing network access costs, and improving the user satisfaction by enabling network control at the edge of the network.

\begin{figure}[t]
\begin{center}
\epsfig{figure=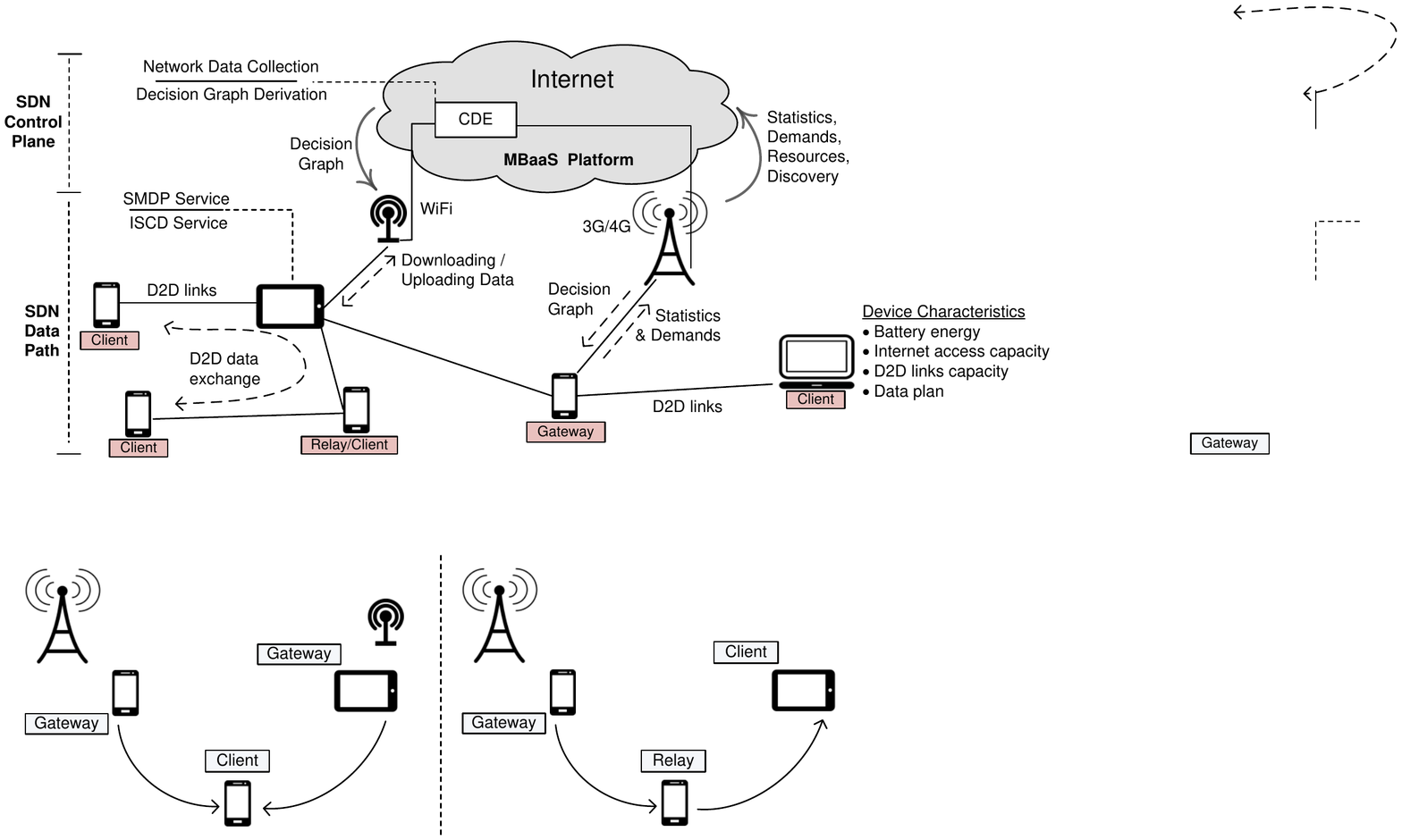, width=8.6cm,height=3.4cm}
\vspace{-1mm}
\caption{An example of users interactions in UPNs. Each user can concurrently consume data from multiple gateways, over multiple and possibly multi-hop paths, and serve as a relay or even a gateway for others. Left: concurrent downloading from two gateways. Right: multi-hop connection to Internet. }
\label{fig:system-model}
\vspace{-3mm}
\end{center}
\end{figure}

One of the first UPN services is FON \cite{fon}, a community-based Wi-Fi Internet access scheme, where roaming FON users can access the Internet through the home Wi-Fi connections of other nearby FON users. Different from FON which utilizes fixed user equipments (Wi-Fi access points), there is an emerging trend of \emph{mobile UPNs}, which focus on leveraging the capabilities of handheld mobile devices \cite{karma}, \cite{opengarden}, \cite{m-87}. For example, Open Garden \cite{opengarden} offers a mobile software which enables mobile devices to connect with each other through Wi-Fi direct \cite{wifi-direct} or Bluetooth \cite{bluetooth}, and share their Internet connections. This solution creates a mesh network, where each user (device) may act as a \emph{client} node (consuming data), a \emph{relay} (relaying data to other nodes), or a \emph{gateway} (connecting the mesh overlay with Internet), as illustrated in Figure \ref{fig:system-model}.

In a nutshell, mobile UPN services aim to crowdsource Internet access by allowing users to collaboratively consume their Internet connections and battery energy. The key idea is to turn the negative externality of network congestion to the positive network effect, by exploiting the diversity of resources and demands of different users. Recent measurement-based studies \cite{laoutaris-conext14} have revealed the large benefits of these sharing mechanisms. \emph{However, the success of such services heavily depends on (i) the willingness of users to join the service and contribute their resources, and (ii) the efficient allocation of the aggregated capacity}. Clearly, a user with low battery energy and fast Internet connection may not be willing to serve other users, unless this improves her satisfaction level. Therefore, it is of important to design a \emph{resource sharing mechanism} for properly incentivizing user participation.

Such an incentive mechanism (currently missing from \cite{opengarden} and similar services) should encourage users to collaborate, and lead to a proper data transmission and routing scheme that balances efficiency and fairness. The \emph{efficiency} is quantified in terms of the aggregate throughput which can be maximized by having users with the highest Internet connection capacities serve as gateways, and by properly scheduling the traffic to avoid heavy interferences among users. The \emph{fairness} criterion, on the other hand, concerns the relationship among data delivery, resource contribution, and economic gains/losses of each user. If a user experiences excessive unfair consumption of her resources she will probably leave the service and hence deteriorate the performance experienced by other users.

Nevertheless, the design of this mechanism is very challenging. First of all, there is often no central entity controlling such a wireless mesh network with devices belonging to heterogeneous physical networks, and each user only has information regarding her own needs and resources. Therefore, the proposed scheme has to be \emph{distributed}. Moreover, the fairness criterion should consider that users have different needs and may contribute different resources with different costs. At the same time, this criterion should take into account that a user will participate in the service only if she expects to improve over her \emph{standalone} performance, i.e., the one she achieves when acting independently.

\vspace{-2mm}
\subsection{Methodology and Contributions}


We introduce a detailed framework for the UPN service, which is modeled as a multi-hop, multi-path mesh network that manages multiple unicast sessions between the Internet and different users. Each device may have one or more network interface cards (NICs), which enable it to transmit/receive over one or more frequency channels. Each user is parameterized by her Internet connection capacity (through cellular or Wi-Fi connections), her available battery energy, the monetary cost for downloading and uploading data (based on the user's pricing scheme or \emph{data plan}), and her relaying capabilities. Finally, we employ user-specific utility and cost functions to account for the communication needs and energy consumption aversion, respectively, that may vary across the users.

\begin{figure}[t]
\begin{center}
\epsfig{figure=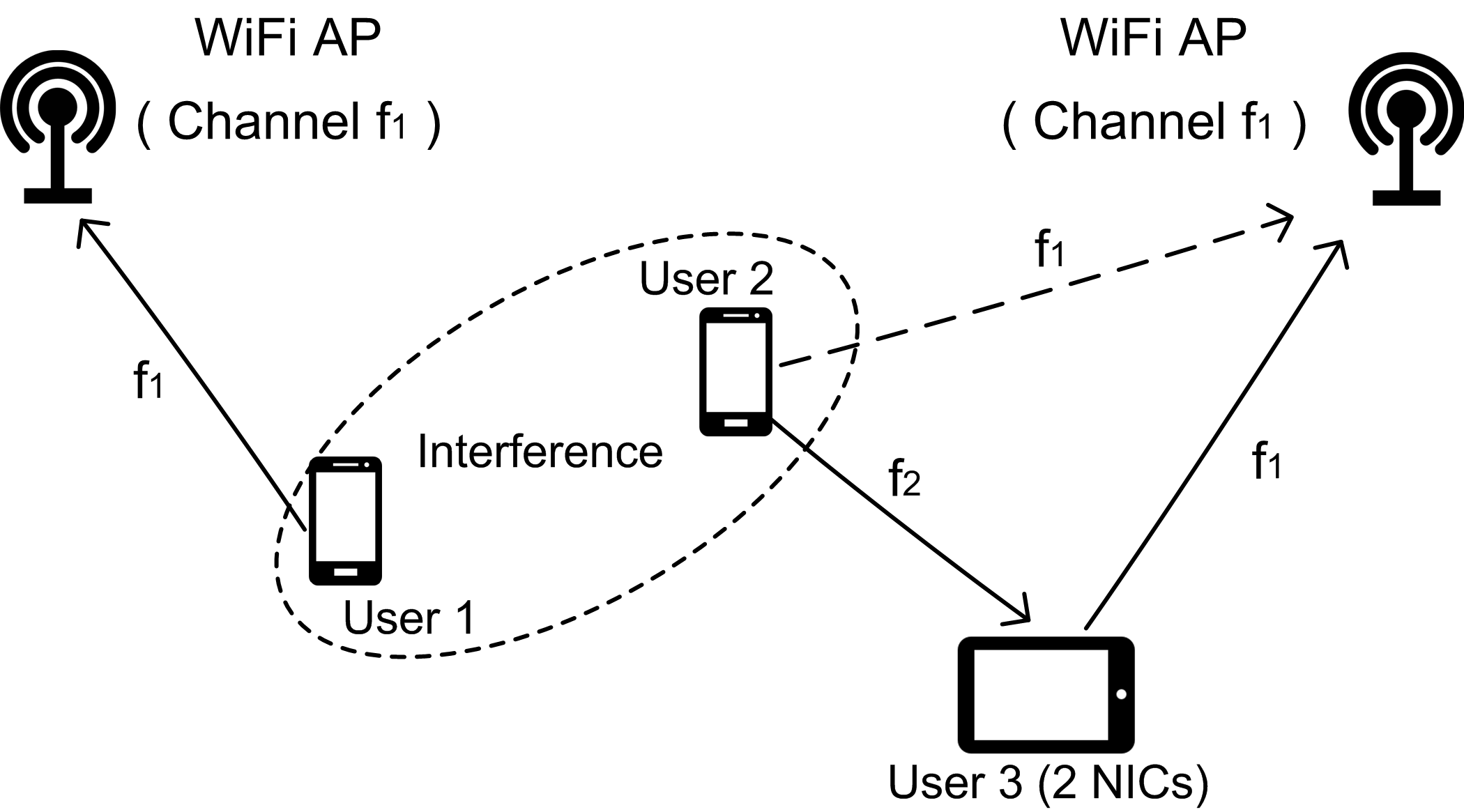, width=6.1cm,height=3.4cm}
\vspace{-2mm}
\caption{The proposed service can be used to mitigate interference by proper channel assignments. In this example, both access points transmit in the same channel (with a frequency $f_1$), hence close-by users ($1$ and $2$) will interfere with each other if they directly communicate with the access points (dashed arrow for user 2). The UPN service can exploit the two NICs of user $3$ (who is further away from user $1$) for relaying the traffic of user 2 through a different channel, thus reducing the interference for user $1$.}
\label{fig:interference}
\end{center}
\vspace{-3mm}
\end{figure}

We use game theory, and specifically the Nash bargaining solution (NBS) concept \cite{nash-bargaining}, to determine the efficient and fair contribution of the user resources and allocation of the service capacity to each user. The NBS yields an outcome which is Pareto efficient and fair \cite{myerson-gametheory}, and hence is self-enforcing, i.e., acceptable by all users. A particularly important feature of the NBS rule is that it takes into account the disagreement performance of each player, i.e., the utility she perceives when an agreement is not reached. This means that the NBS ensures that each user participating in the UPN service will receive a performance at least as good as her standalone performance.

\rev{We further introduce a \emph{virtual currency} system to facilitate the users' cooperation and address the inherent incentive issue. Such a system allows the users to cooperate not only by direct service exchanges (e.g., relaying data for each other) but also by using this currency to pay for the services that they receive (e.g., exchange currency with relayed data)}. This encourages users to participate and serve other users so as to collect currency, even if they currently do not have communication needs (but can use the currency later when they have needs). Similarly, it enables users with poor Internet connections to utilize the service by paying other users. Clearly, this virtual currency system addresses the problem of \emph{double coincidence of needs} \cite{jevons}, therefore increases the number of users who are able to cooperate with each other.

We propose an algorithm that computes the NBS in a distributed fashion, thus enabling the decentralized implementation of the resource sharing mechanism. This is highly non-trivial mathematically, since the respective optimization problem has both a coupled constraint set and a coupled objective function. The obtained solution determines the amount of each resource (Internet access capacity, wireless bandwidth, and battery energy) that each user should contribute, the amount of data that she will be served in return, and the virtual currency transfers. Moreover, it describes how this will be achieved, i.e., how much traffic will be conveyed over each link, which channels will be used for the users' communications, and which Internet connections will be utilized. \rev{Such a holistic methodology enables more than two users to collaborate over multihop paths, and therefore outperforms prior bilateral-only cooperation schemes.} Besides, this joint design and coordination of the users can mitigate interference among users through proper channel reallocations, as shown in Figure \ref{fig:interference}.

Additionally, the proposed service model goes beyond typical self-organized networks or device-to-device communications, as it is capable of connecting multiple heterogeneous Internet access networks. In particular, it can \emph{offload} cellular traffic to Wi-Fi networks \cite{iosifidis-wiopt} or \emph{onload} Wi-Fi traffic to cellular networks \cite{laoutaris-onloading}. These scenarios are illustsrated in Figure \ref{fig:offloading-onloading}. The optimal strategy depends on congestion levels of the Internet connections and the data access costs. More interestingly, these decisions are made by users in a distributed fashion and without the intervention of network operators.

To this end, the main technical contributions are as follows:
\begin{itemize}

\item \emph{Analytical Model}: We introduce a general mobile UPN service that incorporates users' communication needs, monetary costs, and energy consumption, which are the key factors affecting users' servicing decisions.

 \item \emph{Incentive Provision \& Service Allocation Mechanism}: We design a resource sharing mechanism based on the Nash bargaining solution, that induces users' participations through fair allocation of the contributed resources. It is Pareto efficient and takes into account users' standalone performances. These aspects are very crucial to maintain a good performance of the service.

\item \emph{Distributed Algorithm}: We propose a distributed algorithm, which combines the concepts of consistency pricing \cite{chiang-tutorial} and primal-dual Lagrange relaxation \cite{bertsekas-nedic}, and achieves the unique NBS. This enables the decentralized implementation of the service, without requiring central coordination or additional infrastructure.

\item \emph{Intelligence-at-the-edge}: We discuss how the service can account for interference and congestion effects, by taking intelligent (at-the-edge) channel assignment, routing, and flow control decisions. We explain how the service can be used both for mobile data offloading and onloading, adapting on the congestion levels of the different networks as well as the Internet access costs.

\item \emph{Performance Evaluation}: We evaluate the performance of the service for various system parameters and scenarios. We find that the benefits increase as users become more heterogeneous (diverse) in terms of their needs and resources, increasing on average by at least $30\%$ the amount of served data in a typical scenario with $6$ users.

\end{itemize}

The rest of the paper is organized as follows. Section \ref{sec:model} introduces the system model. Section \ref{sec:problem} analyzes the user decisions and provides the problem statement. In Section \ref{sec:mechanism}, we present the Nash bargaining formulation, and in Section \ref{sec:distributed} we provide the algorithm for its distributed solution. We present numerical results in Section \ref{sec:numerical}. Finally, we analyze related works in Section \ref{sec:RelatedWork}, and conclude in Section \ref{sec:conclusions}.

\section{Model}\label{sec:model}

\subsection{Basics of the UPN Model}

\textbf{Network Graph}. We consider a set of mobile users $\mathcal{I}=\{1,2,\ldots,I\}$, who are interested in providing a crowdsourced Internet access service (hereafter referred to as \emph{service}) to each other for a time period $\mathcal{T}$ (e.g., several minutes). The users create a mesh network that is described by a directed graph $G=(\mathcal{I},\mathcal{E}, \mathcal{B})$. Here, $\mathcal{E}$ denotes the set of communication links that can carry data between the nodes (or, users), and $\mathcal{B}$ denotes the set of ``interference links''. If a link $(i,j)\in\mathcal{E}$, then node $i$ can transmit data to $j$. If a link $(i,j)\in\mathcal{B}$, then nodes $i$ and $j$ are not in communication range but still their concurrent transmissions (independent of their transmission destinations) interfere with each other\footnotesc{{The \emph{interference link} is based on \cite{kodialam} and is consistent with the typical interference range approach: an interference link exists among two nodes whenever one is within the interference range of the other. For simplicity, we assume a symmetric relationship.}}. This is possible since in 802.11 the carrier sense range is larger than the transmission range \cite{kodialam}.

We assume that the graph is connected, i.e., any two nodes can communicate (e.g., for exchanging control messages) through a (possibly multi-hop) path along the communication links. We also define the sets
\begin{equation}
In(i)=\{j:(j,i)\in \mathcal{E}\},\,\,Out(i)=\{j:(i,j)\in\mathcal{E}\},
\end{equation}
for the upstream and the downstream one-hop neighbors of node $i$, respectively. Similarly, we introduce the extended sets $In^{e}(i)$ and $Out^{e}(i)$ that include the interference links as well, e.g., $In^{e}(i)=\{j:(i,j)\in\mathcal{E}\cup\mathcal{B}\}$. {Finally, we define the set of neighbors for $i$, $\mathcal{N}(i)=In(i)\cup Out(i)$ and the respective extended neighbor set $\mathcal{N}^{e}(i)=In^{e}(i)\cup Out^{e}(i)$.}

\begin{figure}[t]
	\begin{center}
		\epsfig{figure=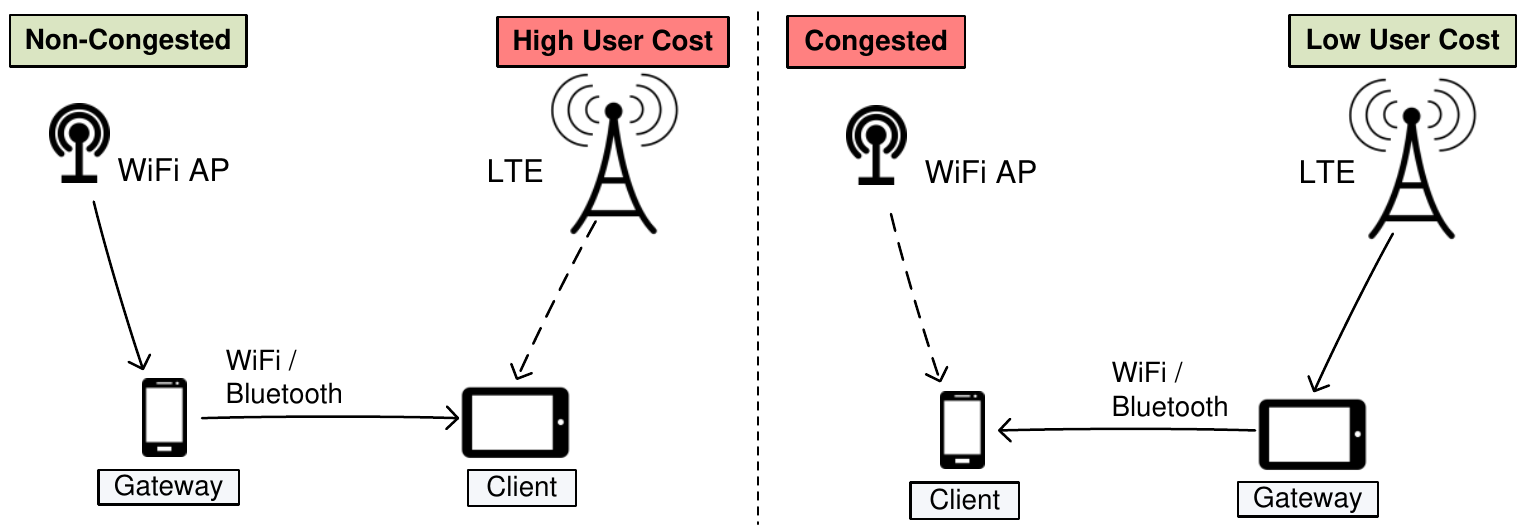, width=8.2cm,height=3.3cm}
		\caption{The proposed service can be used to \emph{offload} cellular traffic to Wi-Fi access points (left), or \emph{onload} Wi-Fi traffic to cellular networks (right). The offloading decisions depend mainly on the data usage (and energy) cost incurred by the users, while the onloading decisions depend on the congestion of Wi-Fi APs. Here the dashed arrow means the data downloading choice before collaboration, and the solid arrow means the data downloading choice after collaboration.}
		\label{fig:offloading-onloading}
	\end{center}
	\vspace{-2mm}
\end{figure}

\textbf{Channels}. Each node can access the Internet using multiple Internet connections simultaneously, for example through both cellular and Wi-Fi connections. This is consistent with the recent technology development (e.g., see \cite{keown}, \cite{syrivelis}). Our model also allows the more restricted case where each gateway node can only use the best of her multiple Internet connections.

There is a set $\mathcal{F}$ of unlicensed channels (Bluetooth or Wi-Fi) available for each direct link among the users\footnotesc{Note that the Wi-Fi Internet access channel is determined by the Wi-Fi AP instead of by users participating in the UPN. Hence, it may or may not coincide with one of the orthogonal channels in $\mathcal{F}$. For the more general model, we assume that this channel is one of the elements in set $\mathcal{F}$.}. For example, in the IEEE 802.11 family of standards \cite{wifi-direct}, there are $3$ orthogonal channels in the 2.4GHz band and 12 channels in the 5GHz band. Each communication link can utilize any one (and only one) of these channels $f\in\mathcal{F}$ at any given time. Note that set $\mathcal{F}$ does \emph{not} include the licensed cellular channel which we denote as $c$. Each node $i\in\mathcal{I}$ has $k_i$ network interface cards (NICs) which allow the concurrent utilization of up to $k_i$ different channels on $k_i$ communication links connected to this node. This includes the special case of $k_i=1$, where a node has only one Wi-FI NIC (besides her capability of accessing the Internet through the cellular channel $c$). Similarly to \cite{kodialam}, \cite{XLin}, we assume that the devices can change channels very fast without significant overheads (e.g., in energy consumption).

\textbf{Link Capacities and Commodities}. We denote with $C_{i}^{c}$ and $C_{i}^{f}$ the Internet access capacity for user $i\in\mathcal{I}$ when she employs her cellular connection $c$, or Wi-Fi connection choosing channel $f\in\mathcal{F}$, respectively. Also, let $C_{ij}^{f}\geq 0$ denote the capacity of link $(i,j)\in\mathcal{E}$ when choosing channel $f\in\mathcal{F}$ assuming no other interference links are active at the same time. In general, we allow to have $C_{ij}^{f_1}\neq C_{ij}^{f_2}$ for $f_1,f_2\in\mathcal{F}$, so as to take into account possible channel diversity. We assume that the users are nomadic (not very fast moving), and the time period length much larger than the channel coherent time. This means that we can use proper coding techniques to deal with the fast small scale fading, therefore the link capacities are considered constant during $\mathcal{T}$. 

\vspace{-2mm}
\subsection{Network Constraints, Routing and Internet Access}

In order to simplify the presentation, we will focus on the scenario where users download data from the Internet. Our analysis can be easily extended for the uploading scenario with more complicated notations. \rev{During period $\mathcal{T}$, each user can have more than one roles if needed: she can be a \emph{client} node (consuming data), a \emph{relay} node (routing data to/from other users), and/or a \emph{gateway} node (downloading data from the Internet)}. Therefore, there exist at most $|\mathcal{I}|$ data commodities in the UPN, where each commodity $n\in\mathcal{I}$ corresponds to a (potentially multi-hop) unicast session originating from the Internet (e.g., a content server) and ending at a user $n\in\mathcal{I}$.

\textbf{Scheduling and Network Constraints}. We assume that the time period $\mathcal{T}=\{1,2,\ldots,T\}$ is divided in $T$ equal length time slots $t=1,2,\ldots,T$, with each slot having a normalized length of $1$. In each slot, \rev{every node $i\in\mathcal{I}$ can potentially update her strategy and decide whether she is going to use any of her outgoing links to transmit data, or receive data from her upstream neighbors, or access the Internet through the cellular or a Wi-Fi network}. If the node has multiple NICs, she can choose multiple choices at the same time. Namely, we use the binary variable $\widehat{x}_{ij}^{f}[t]\in\{0,1\}$ to denote node $i$'s decision of whether transmitting data (of any commodity) to node $j\in Out(i)$ during slot $t$ in channel $f$. We further use binary variables $\widehat{y}_{i}^{c}[t]\in\{0,1\}$ and $\widehat{y}_{i}^{f}[t]\in\{0,1\}$ to denote user $i$'s decisions of whether downloading data from the Internet during $t$, through the cellular connection or the Wi-Fi channel $f\in\mathcal{F}$, respectively.

These decisions are subject to certain constraints. The first condition that should be satisfied is the \emph{link-channel} constraint: for each link $(i,j)\in\mathcal{E}$ and each Wi-Fi Internet link, only one channel can be employed in each slot $t$ at most, i.e.,
\begin{align}
\sum_{f\in\mathcal{F}}\widehat{x}_{ij}^{f}[t]\leq 1,\,\,\, \forall\,(i,j)\in\mathcal{E},\,\,\forall\, t\in\mathcal{T}, \label{eq:neces-1-1} \\
\sum_{f\in\mathcal{F}}\widehat{y}_{i}^{f}[t]\leq 1,\,\,\, \forall\,i\in\mathcal{I},\forall\, t\in\mathcal{T}\,. \label{eq:neces-1-2}
\end{align}
\noindent Second, the data amounts that each node $i\in\mathcal{I}$ can send or receive from her neighbors or from Internet, in all the available channels, are constrained by the total number of her network interface cards. Therefore, the following \emph{node constraint} should hold for every node $i\in\mathcal{I}$ and slot $t\in\mathcal{T}$:
\begin{equation}
\sum_{f\in\mathcal{F}}\left( \widehat{y}_{i}^{f}[t] + \sum_{k\in In(i)}\widehat{x}_{ki}^{f}[t]+\sum_{j\in Out(i)}\widehat{x}_{ij}^{f}[t] \right)\leq k_i,\, \label{eq:neces-2}
\end{equation}
which does not include the cellular transmissions since they use a different NIC.

Moreover, the transmissions of the different nodes are coupled due to interference. In particular, based on the protocol interference model \cite{kodialam}, \cite{gupta}, we assume that a transmission over link $(i,j)\in\mathcal{E}$ in $t$ is successful only if all nodes that are connected to $i$ or $j$ through an interference or communication link are idle, i.e., do not transmit or receive data during $t$. To facilitate modeling, we define the set of all interfering links:
\begin{align}
I(i,j)=&\{ (c,b),(a,c): c\in \mathcal{N}^{e}(i)\cup \mathcal{N}^{e}(j), \nonumber \\
&b\in Out(c),\, a\in In(c) \},\,\forall (i,j)\in\mathcal{E}\,. \nonumber
\end{align}
Therefore we have the set of interference constraints: 
\begin{align}
&\widehat{x}_{ij}^{f}[t]+\sum_{(k,m)\in I(i,j)}\widehat{x}_{km}^{f}[t]\, + \nonumber \\
&\sum_{k\in \mathcal{N}^{e}(i)\cup \mathcal{N}^{e}(j)}\widehat{y}_{k}^{f}[t]\leq 1,\,\forall (i,j)\in\mathcal{E}, t\in\mathcal{T}, f\in\mathcal{F}  \label{eq:neces-3}
\end{align}

Equation sets (\ref{eq:neces-1-1})-(\ref{eq:neces-3}) constitute the \emph{sufficient} and \emph{necessary} conditions for a feasible channel assignment and scheduling in the UPN during $\mathcal{T}$. However, the above binary variables render any scheduling optimization problem highly complex since they require - as it will be explained below - the solution of challenging discrete non-linear optimization problems. To address this issue, we employ a continuous-time approximation for modeling the operation of the UPN. \rev{This yields the one-off servicing policy the nodes decide at the beginning of each period $\mathcal{T}$. Nevertheless, since this policy is based on the detailed slot-by-slot analysis, it enables the design of near-optimal solutions as it will become clear next.}

\textbf{Routing and Internet Access Decisions}. Specifically, we follow the analysis (among others) of \cite{kodialam}, \cite{XLin}, and \cite{yuan-ups}, and relax the above discrete decision variables by employing the average data rates over each link in every channel. In detail, we use $y_{i}^{c}(n)\geq 0$ and $y_{i}^{f}(n)\geq 0$ to denote user $i$'s average downloading rate (from the Internet) for commodity $n\in\mathcal{I}$, over the cellular connection or the Wi-Fi channel $f\in\mathcal{F}$, respectively. Specifically, we define    :
\begin{align}
&\sum_{n\in\mathcal{I}}y_{i}^{f}(n)=C_{i}^{f}\sum_{t=1}^{T}\widehat{y}_{i}^{f}[t]/T\,,\,\,\forall\,i\in\mathcal{I},\,f\in\mathcal{F}\,, \label{eq:discrete-to-cont-yf} \\
&\sum_{n\in\mathcal{I}}y_{i}^{c}(n)=C_{i}^{c}\sum_{t=1}^{T}\widehat{y}_{i}^{c}[t]/T\,,\,\, \forall\,i\in\mathcal{I}, \label{eq:discrete-to-cont-yc}
\end{align}
where notice that we sum over all commodities.

Also, we denote $x_{ij}^{f}(n)\geq 0$ as user $i$'s average rate of routing data of commodity $n$ to her one-hop downstream neighbor $j$, using channel $f\in\mathcal{F}$. These routing decisions are determined by the scheduling decisions:
\begin{equation}
\sum_{n\in\mathcal{I}}x_{ij}^{f}(n)=C_{ij}^{f}\sum_{t=1}^{T}\widehat{x}_{ij}^{f}[t]/T,, \,\forall (i,j)\in\mathcal{E},\,f\in\mathcal{F} \label{eq:discrete-to-cont-xij}
\end{equation}
Using the above definitions we can describe the operation of the UPN with the \emph{routing} matrix $\bm{x}=\big(x_{ij}^{f}(n)\geq 0:(i,j)\in\mathcal{E},n\in\mathcal{I},n\neq i,f\in\mathcal{F}\big)$, and the \emph{Internet access} matrix $\bm{y}=\big(y_{i}^{f}(n),y_{i}^{c}(n)\geq 0:i\in\mathcal{I},n\in\mathcal{I},f\in\mathcal{F}\big)$. These matrices comprise also the channel and commodity selection decisions.

With the above continuous-time relaxation and substituting (\ref{eq:discrete-to-cont-yf})-(\ref{eq:discrete-to-cont-xij}) to (\ref{eq:neces-1-1})-(\ref{eq:neces-3}) we devise the necessary conditions for a schedule to be feasible. In detail, based on (\ref{eq:neces-2}) we have the following \emph{node-radio} constraint set for the Wi-Fi NICs:
\begin{align}
&\sum_{k\in In(i)}\sum_{f\in\mathcal{F}}\frac{\sum_{n\in\mathcal{I}}x_{ki}^{f}(n)}{C_{ki}^{f}} + \sum_{j\in Out(i)}\sum_{f\in\mathcal{F}}\frac{\sum_{n\in\mathcal{I}}x_{ij}^{f}(n)}{C_{ij}^{f}} \nonumber \\
&+ \sum_{f\in\mathcal{F}}\frac{\sum_{n\in\mathcal{I}}y_{i}^{f}(n)}{C_{i}^{f}}    \leq k_i,\,\forall\,i\in\mathcal{I}. \label{eq:node-radio-constraint}
\end{align}
Similarly, eq. (\ref{eq:neces-3}) lead to the scheduling constraints:
\begin{align}
&\sum_{(k,m)\in I(i,j)}\frac{\sum_{n\in\mathcal{I}}x_{km}^{f}(n)}{C_{km}^{f}}+ \frac{\sum_{n\in\mathcal{I}}x_{ij}^{f}(n)}{C_{ij}^{f}} + \label{eq:interference-constraint} \\
&\sum_{k\in \mathcal{N}^{e}(i)\cup \mathcal{N}^{e}(j)}\frac{\sum_{n\in\mathcal{I}}{y}_{k}^{f}(n)}{C_{k}^{f}} \,\, \leq 1,\,\forall\,(i,j)\in\mathcal{E},\,\forall f\in\mathcal{F}\,. \nonumber
\end{align}
Note that the last term includes the Internet access transmissions of both nodes $i$ (as a neighbor of $j$) and $j$ (as a neighbor of $i$). Moreover, according to (\ref{eq:neces-1-1}) and (\ref{eq:neces-1-2}), the transmissions in different channels over each link are coupled:
\begin{equation}
\sum_{f\in\mathcal{F}}\frac{\sum_{n\in\mathcal{N}}x_{ij}^{f}(n)}{C_{ij}^{f}}\leq 1,\,\, \forall (i,j)\in\mathcal{E}, \label{eq:wifi-congestion}
\end{equation}
\begin{equation}
\sum_{f\in\mathcal{F}}\frac{\sum_{n\in\mathcal{N}}y_{i}^{f}(n) }{C_{i}^{f}}\leq 1,\,\,\,\,\,\,\sum_{n\in\mathcal{I}}y_{i}^{c}(n)\leq C_{i}^{c},\,\,\forall\,i\in\mathcal{I}. \label{eq:cellular-congestion}
\end{equation}

Finally, the routing and Internet access variables should satisfy the flow conservation constraints:
\begin{align}
&\sum_{j\in In(i)}\sum_{f\in\mathcal{F}}x_{ji}^{f}(n)+ y_{i}^{c}(n)+\sum_{f\in\mathcal{F}} y_{i}^{f}(n)= \nonumber \\
&\sum_{j\in Out(i)}\sum_{f\in\mathcal{F}}x_{ij}^{f}(n),\,\,\forall\,i,\,n\in\mathcal{I},\,n\neq i. \label{eq:flow-conserve101} 
\end{align}
In other words, the data that node $i\in\mathcal{I}$ receives from her upstream neighbors plus the data she downloads (over all channels) should be equal to the amount of data she routes to her downstream neighbors. This should hold for all commodities except $n=i$ that is consumed locally by $i$.

Using matrices $\bm{x}$ and $\bm{y}$, we can determine the UPN operation (based on the objective that we will define in the sequel), which is implementable (feasible) but not necessarily optimal. In other words, due to the continuous-time relaxation the above constraints are necessary but not sufficient for optimality \cite{kodialam}. On the other hand, we note that the proposed scheme is expected to operate based on the 802.11 MAC protocol, which does not support synchronous operation and hence cannot apply the discrete-time solution (even if it was known). A detailed gap characterization between the theoretical performance of the proposed scheme and its implementation (e.g., by a practical protocol) will be part of our future work.


\section{Users Decisions and Problem Statement}\label{sec:problem}

\textbf{Data Consumption Utility}. Each user $i\in\mathcal{I}$ has elastic needs and perceives certain satisfaction for consuming data (not including relaying for other users). This is modeled by a utility function $U_{i}(\cdot)$ that depends on the aggregate average rate $r_i$, with which user $i$ directly downloads and receives data of commodity $i$ from her neighbors, i.e.,
\begin{equation}
r_i=y_{i}^{c}(i) + \sum_{f\in\mathcal{F}}y_{i}^{f}(i) + \sum_{j\in In(i)}\sum_{f\in\mathcal{F}}x_{ji}^{f}(i)\,. \label{eq:r-i-function}
\end{equation}
Function $U_i(\cdot)$ is assumed to be positive, increasing, and strictly concave. The concavity captures the user's diminishing marginal satisfaction of additional data consumption. Different users may have different utility functions \cite{lui-mesh}, \cite{walrand-wifi}. For example, the utility of a user who is streaming a video file is initially proportional to the downloading rate, and saturates after the maximum available (encoding) rate has been reached. On the other hand, the utility of a user who is downloading a file increases strictly with the downloading rate.

\textbf{Energy Expenditure Cost}. Energy consumption is a major consideration for mobile users since their devices have often limited energy resources \cite{iphone}. Let $e_{ij}^{f,s}$ (joules/bytes) be the energy that node $i$ consumes when she sends one byte to node $j$ over link $(i,j)\in\mathcal{E}$ in channel $f$. Also, $e_{ij}^{f,r}$ denotes the energy that node $j$ consumes for receiving one byte from node $i$ in $f$. Finally, $e_{i}^{c}$ and $e_{i}^{f}$ are the energy consumptions when node $i$ downloads one byte from the Internet, through her Wi-Fi or cellular connection, respectively.

\rev{
Providing analytical expressions for wireless transmission energy costs is particularly challenging (e.g., see seminal work \cite{nilsson}), and is further perplexed in the case of handheld devices. For example, the relation of power consumption with throughput is affected by the method that is employed to increase the latter \cite{koutsonikolas-infocom15}. Nevertheless, it is commonly agreed that the energy consumption depends on the transmission time and the volume of transmitted data, and therefore on the link capacity which in turn is shaped by the channel conditions. To capture qualitatively the above aspects and avoid delving into the physical layer details, we follow here the measurement studies \cite{balasubra-imc}, \cite{huang-mobisys}, which indicate that the power consumption (in mWatts) for transmitting ($P_{ij}^{f,s}$) and receiving ($P_{ij}^{f,r}$) over link $(i,j)\in\mathcal{E}$ has a constant offset plus a term linearly dependent on the rate (hence depends also on channel $f$)}:
\begin{equation}
P_{ij}^{f,s}=\delta^{f,s}C_{ij}^{f}+\theta^{f,s},\,\,\,P_{ij}^{f,r}=\delta^{f,r}C_{ij}^{f}+\theta^{f,r}\nonumber
\end{equation}
Parameters $\delta^{f,s}$, $\delta^{f,r}$, $\theta^{f,s}$, and $\theta^{f,r}$ are constants and depend on the technology choices (Wi-Fi, 3G, or LTE-A) \cite{balasubra-imc}, \cite{huang-mobisys}. \rev{Hence, the energy consumption of node $i$ in link $(i,j)$ is\footnotesc{\rev{This formulation assumes that each time a link is employed, it can achieve its maximum capacity since the underlying scheduling scheme ensures a proper coordination among the interfering transmissions \cite{tdm-mac}, \cite{XLin}. If the scheduling is achieved with the traditional CSMA/CA protocol, i.e., not a time-slotted system, then the energy consumption will be higher due to collisions.}}:
\begin{equation}
\sum_{n\in\mathcal{I}}\sum_{f\in\mathcal{F}}e_{ij}^{f,s}x_{ij}^{f}(n)T=\sum_{n\in\mathcal{I}}\sum_{f\in\mathcal{F}}\frac{x_{ij}^{f}(n)}{C_{ij}^{f}}T\big(\delta^{f,s}C_{ij}^{f}+\theta^{f,s}\big).
\label{eq:energy-equation001} \nonumber
\end{equation}
Note that the first term on the RHS captures the total time (within $\mathcal{T}$) used for transmitting data in $f$, and the second term of this product captures the power consumption on a particular data rate.} The energy consumption for a receiver can be defined similarly. Finally, the total energy consumption by node $i$ is:
\begin{align}
&e_{i}= \sum_{n\in\mathcal{I}}\sum_{f\in\mathcal{F}}e_{i}^{f}y_{i}^{f}(n)T + \sum_{j\in Out(i)}\sum_{n\in\mathcal{I}}\sum_{f\in\mathcal{F}}e_{ij}^{f,s}x_{ij}^{f}(n)T \nonumber \\
&+ e_{i}^{c}\sum_{n\in\mathcal{I}}y_{i}^{c}(n)T + \sum_{j\in In(i)}\sum_{n\in\mathcal{I}}\sum_{f\in\mathcal{F}}e_{ji}^{f,r}x_{ji}^{f}(n)T\,.
\label{eq:energy-equation002}
\end{align}

We assume that each node $i\in\mathcal{I}$ has a maximum energy budget of $E_i\geq 0$ units (joules). This parameter can be explicitly set by the user (e.g., through a proper user interface on the mobile app), or it may represent her actual residual battery. Clearly, we need to have $e_{i}\leq E_i$ for each node. Moreover, different users may have different energy consumption preferences. For example, some users may be willing to consume almost their entire energy budgets, while others may prefer a much more conservative energy consumption.

To capture the above user difference, we introduce for each user $i$ an energy cost (or, \emph{dissatisfaction}) function $V_i(\cdot)$, which is positive, increasing, and strictly convex in $e_{i}$. Its value goes to infinity when the energy budget of the user is depleted. A function that satisfies these requirements is, for example, $V_i(e_{i})=\phi_i/(E_i-e_{i})$, where $\phi_i\in [0,1]$ is a normalization parameter indicating user $i$'s sensitivity in energy consumption. \rev{Finally, we wish to stress that there is an additional energy cost for the exchange of coordination messages (as it will be explained in Section \ref{sec:distributed}) and an energy consumption due to ACK messages. The latter is a relatively small portion of the relayed traffic volume, and does not have a significant impact on the UPN operation \cite{syrivelis}.} 

\rev{
\textbf{Data Plan Cost}. The impact of the remaining data plan quota on users' collaboration decisions are very crucial. Yet, the analysis of quota dynamics is very challenging \cite{andrews-journal}, especially if the demand is elastic and dynamically decided (as here). In order to study this aspect in a tractable fashion, we characterize each user by a \emph{psychological price} per unit of cellular data, as a function of her remaining data quota. 

We assume that these psychological parameters remain fixed during each time period, but they can change across different time periods so as to reflect the quota's impact. Therefore, the dissatisfaction of each user $i$ due to the consumption of her data plan can be described by a convex function as follows:
\begin{equation}
Q_i(y_{i}^{ag,c})=\frac{o_i}{A_i-y_{i}^{ag,c}} \label{eq:quota}
\end{equation}
where $y_{i}^{ag,c}=\sum_{n\in\mathcal{I}}y_{i}^{c}(n)T$ is the aggregate amount of data that will be downloaded through her cellular link during $\mathcal{T}$, and $A_i$ the available quota at the beginning of the period. Parameter $o_i>0$ is the psychological price, which captures the user aversion on the consumption of her data plan (including the impact of the actual cost per byte, the anticipation of future needs, and other such latent factors) and can vary both across users and time periods. This model has the following desirable properties. First, as the amount of consumed data approaches the currently available quota, the psychological monetary cost increases very fast. Second, even if two users $i$ and $j$ have the same quotas ($A_i=A_j$) and download equal amounts of data ($y_{i}^{ag,c}=y_{j}^{ag,c}$), they still might incur different costs due to their preferences ($o_i\neq o_j$).}

\textbf{Payoff Function}. In order to capture all the above aspects that affect the users' decisions, we introduce the payoff function $J_{i}(\cdot)$ that user $i$ receives when she participates in the crowdsourced Internet access service. \rev{More specifically:
\begin{align}
J_{i}(\bm{x}_i,\bm{x}^{in}_{i},\bm{y}_{i})&=U_i(r_i) - Q_i(y_{i}^{ag,c}) \nonumber \\
 &- \sum_{f\in\mathcal{F}}p_{i}^{f}\sum_{n\in\mathcal{I}}y_{i}^{f}(n)T -  V_i(e_{i}), \nonumber
\end{align}
where $p_{i}^{f}\geq 0$ is the price for accessing Internet through Wi-Fi channel $f$.} Also, we have defined the matrices of downloading and routing: $\bm{y}_i=(y_{i}^{f}(n),\,y_{i}^{c}(n):\,n\in\mathcal{I},f\in\mathcal{F})$, $\bm{x}_i=(x_{ij}^{f}(n):\,j\in Out(i),\,n\in\mathcal{I}, f\in\mathcal{F})$, and $\bm{x}^{in}_{i}=(x_{ji}^{f}(n):\,j\in In (i),\,n\in\mathcal{I}, f\in\mathcal{F})$.


\begin{table}
	\centering%
	\begin{tabular}{|c|c|}
		\hline %
		\hline
		Symbol & Physical Meaning\\
		\hline %
		$C_{ij}^{f}$ & Capacity of link $(i,j)\in\mathcal{E}$, channel $f$\\
		\hline
		$k_{i}$ & Total number of NICs of each node $i\in\mathcal{I}$\\
		\hline
		$p_{i}^{f}$ & Cost per  byte over the Wi-Fi Internet link of $i\in\mathcal{I}$\\
		\hline
		$C_{i}^{f}$ & Capacity of Internet access link of node $i$ in ch. $f$ (Wi-Fi)\\
		\hline %
		$C_{i}^{c}$ & Capacity of Internet access cellular link\\
		\hline %
		$e_{ij}^{f,s}$ & Consumed energy per transm. byte over link $(i,j)$ (ch. $f$)\\
		\hline %
		$e_{ij}^{f,r}$ & Consumed energy per receiv. byte over link $(i,j)$ (ch. $f$)\\
		\hline %
		$e_{i}^{f}$ & Consumed energy per receiv. byte at $i$ for Wi-Fi Internet ($f$)\\
		\hline %
		$e_{i}^{c}$ & Consumed energy per receiv. byte at $i$ for cellular Internet\\
		\hline %
		$x_{ij}^{f}(n)$ & Average transfer rate of commodity $n$ over $(i,j)$ in ch. $f$\\
		\hline
		$y_{i}^{f}(n)$ & Average download rate of commodity $n$ by user $i$ (ch. $f$)\\
		\hline
		$y_{i}^{c}(n)$ & Average download rate of commodity $n$ by user $i$ (cellular)\\
		\hline
		$z_{ij}(n)$ & Currency that user $j$ pays to $i$ for receiving data $n$\\
		\hline
		$\bm{x}$ & Routing matrix\\
		\hline
		$\bm{y}$ & Internet access matrix\\
		\hline
		$U_{i}(\cdot)$ & Utility function for user $i$\\
		\hline
		$V_{i}(\cdot)$ & Energy consumption cost function for user $i$\\
		\hline
		$Q_{i}(\cdot)$ & \rev{Data plan cost function for user} $i$\\
		\hline
		$J_{i}(\cdot)$ & Payoff function for user $i$ when joining the service\\
		\hline
		$J_{i}^{s}$ & Standalone performance for user $i$\\
		\hline
		$\mathcal{T},\, T$ & Time period $\mathcal{T}$ comprising $T$ slots of unit length\\
		\hline
	\end{tabular}
	\label{table:notation} 
\end{table}

Note that the payoff $J_{i}$ monotonically decreases with the amount of data that user $i$ downloads for other users $n\in\mathcal{I},\,n\neq i$, and routes to her downstream neighbors. \rev{Therefore, users are reluctant to execute these tasks without proper compensations. Furthermore, some users may not have communication needs in a certain time period, and therefore may not be willing to participate in the UPN service. To address these issues, it is necessary to use a payment mechanism that will alter the payoff function of the users and promote collaboration. We provide the details in Section \ref{sec:mechanism}.}

\textbf{Standalone Operation}. Moreover, a rational user will join the service only if this will improve her payoff in comparison to her standalone performance. In the standalone operation, each user $i$ does not receive or deliver data to her neighbors, nor she downloads data for any other user. Therefore, the optimal Internet access strategy can be obtained by solving the following \emph{Standalone Operation Problem} (SOP):
\begin{align}
\max_{y_{i}^{c}(i),\{y_{i}^{f}(i)\}_{f\in\mathcal{F}}}\,\, &U_i\big( y_{i}^{c}(i) + \sum_{f\in\mathcal{F}}y_{i}^{f}(i)\big)-p_{i}^{f}\sum_{f\in\mathcal{F}}y_{i}^{f}(i)T\nonumber \\
& - Q_i(y_{i}^{c}) - V_i\big( y_{i}^{c}(i) + \sum_{f\in\mathcal{F}}y_{i}^{f}(i) \big) \label{eq:sop}
\end{align}
\begin{equation}
\texttt{s.t.}\,\,\, \sum_{f\in\mathcal{F}}\frac{y_{i}^{f}(i)}{C_{i}^{f}}\leq 1,\,\,\,\,0\leq y_{i}^{c}(i)\leq C_{i}^{c}\,,
\end{equation}
where we have written both $U_i(\cdot)$ and $V_i(\cdot)$ as functions of the downloading decisions of user $i$. This problem has a strictly concave objective, and a compact and convex non-empty constraint set (under the assumption of at least one non-zero Internet access capacity). Hence, it has a unique solution denoted as $J_{i}^{s}$, where $s$ stands for ``standalone''. This will serve in the sequel as the performance benchmark for the comparison purpose.

\textbf{Problem Statement}. We are interested in designing a resource sharing mechanism that determines how much resources each user should contribute, in terms of Internet access, relaying bandwidth, and battery energy, so as to maximize the service capacity (amount of data delivered within period $\mathcal{T}$). Accordingly, the mechanism should decide how this capacity will be shared by the users, and how they should be compensated for their contribution in terms of money transfers. These tasks should be jointly designed to satisfy the agreed fairness criterion. Formally, the problem is defined as follows:

\vspace{1mm}
\text{}\textbf{UPN Collaboration and Servicing Problem:} \emph{Given the graph $G=(\mathcal{I},\mathcal{E}, \mathcal{B})$, the capacity constraints, energy consumption parameters, pricing parameters, and the users' utility functions, find the Internet access, routing and payment decisions of the users, which ensure the fair and efficient performance of the crowdsourced Internet access service}.

\section{The Cooperative Servicing Game}\label{sec:mechanism}

\subsection{Virtual Currency System}

In these cooperative systems, an important issue that may deteriorate their performance is the problem of \emph{double coincidence of needs} \cite{jevons}. \rev{In the context of UPNs, this problem appears as follows. A user served by another user may not be able to directly return the favor in the current period by offering similar services. Therefore, users may not want to help those that cannot reciprocate. Clearly, this problem reduces the number of users who can potentially collaborate with each other and impacts the UPN performance.  

To address the above issue, we introduce a virtual currency system where users pay for the services they receive and are paid when offering such services. This solution makes offering service more attractive even for users who do not currently have communication needs.} Since we aim at a decentralized service design, we assume that transactions are only possible among adjacent nodes (i.e., two nodes of the same link) instead of between the final destination node and the gateway or the intermediate relays. Specifically, let $z_{ji}(n)\geq 0$ denote the currency paid by user $i$ to $j$, for the data of commodity $n$ that is delivered over link $(j,i)\in\mathcal{E}$.

At the beginning of the current period, each user $i$ has a budget $D_i\geq 0$, and is rewarded with an additional amount $\gamma >0$ of virtual currency by the system for her participation in the current period. This value of $\gamma$ is small, compared to the payments exchanged among the nodes, and it is identical for each user. However, $\gamma$ is very important, as we explain in detail below, since it ensures that all users are willing to participate in the service even if they don't eventually exchange Internet access and relaying services with the other users.

At the end of the period, user $i$'s virtual currency budget is:
\begin{eqnarray}
H_i(\bm{z}_i,\bm{z}^{out}_{i})&=& \beta_i\biggl(D_i+ \gamma+ \sum_{n\in\mathcal{I}}\sum_{j\in Out(i)}z_{ij}(n) \nonumber \\
&-&  \sum_{n\in\mathcal{I}}\sum_{j\in In(i)}z_{ji}(n)\biggr), \nonumber
\end{eqnarray}
where we have defined $\bm{z}_i=\big(z_{ji}(n):j\in In(i),n\in\mathcal{I}\big)$ and $\bm{z}^{out}_{i}=\big(z_{ij}(n):j\in Out(i),n\in\mathcal{I}\big)$. Parameter $\beta_i>0$ captures how important the virtual currency is for user $i$, i.e., reflects her expectation for exploiting the virtual currency in the future. Parameter $\beta_i$ can be considered as the \emph{discount rate} for each user. For example, a user that does not intend to participate in the service later does not value the virtual currency much, and her corresponding $\beta_i$ will be close to $0$. Alternatively, these parameters can be set by the service so as to normalize the virtual currency benefit with the benefits of the served data. The linear form of $H_{i}(\bm{z}_i,\bm{z}^{out}_{i})$ implies that users are risk neutral \cite{myerson-gametheory}. After introducing the virtual currency system, the payoff of each user becomes the sum of $J_{i}(\cdot)$ and $H_{i}(\cdot)$.

\vspace{-3mm}
\subsection{Bargaining Problem}

The users are self-interested, and only participate in the crowdsourced connectivity service if this ensures higher payoffs for them. In this work, we propose a resource sharing scheme based on the Nash bargaining solution (NBS), which has the following desirable properties \cite{myerson-gametheory}, \cite{mazumdar}:

\emph{Strong Efficiency}: The solution is Pareto optimal. Hence, there is no other feasible solution which yields a better payoff than the NBS for one user, and no worse payoff for all the other users.

\emph{Individual Rationality}: The solution offers to each user a payoff that is no worse than the payoff she has when she does not participate in the bargaining game (\emph{disagreement point}). This property is especially important for our problem, as a fairness rule based on direct resource allocation only (e.g., an equal energy or bandwidth sharing scheme) may fail to incentivize all users to join the service.

\emph{Scale Covariance}: If we apply an affine transformation to the players' utility functions, then the initial solution can yield the corresponding bargaining solution for the new problem under the same affine transformation. This property ensures that the resource allocation is fair and optimal, independently of the way that we measure the players' utilities.

\emph{Symmetry and Independence of Irrelevant Alternatives}: These two properties ensure that outcomes which would not have been selected (i.e., they are not favorable to users), do not affect the bargaining solution. Moreover, the payoff that each user receives does not depend on her identity/label.

We define an $|\mathcal{I}|$-person bargaining game \cite{nash-bargaining}, where users can exchange services and pay each other with virtual currency. Hence, the produced welfare (service capacity and currency) can be divided in an arbitrary fashion among the users. In line with similar commercial services, e.g., \cite{opengarden}, we assume that when a user joins the UPN she may cooperate with any of the other nearby users, i.e., there is no option for selecting with whom to cooperate. Hence, this is a pure bargaining problem.

However, the formulation and distributed algorithm design for solving this NBS problem are highly non-trivial and depart significantly from previous related approaches, e.g., \cite{mazumdar}. Namely, the coupling of the decisions of different users in their payoff functions (i.e., they need to agree on channels and allocated rates per commodity) and in the constraints set (e.g., when computing the link capacities), as well as the existence of the virtual currency payments call for a new approach.

\begin{figure*}[!t]
	\normalsize
	\setcounter{MYtempeqncnt}{\value{equation}}
	\setcounter{equation}{33}
	\small
	\begin{align}
	&L= \sum_{i\in\mathcal{I}}\big[\log \big(J_{i}^{c}(\bm{x}_i, \bm{\xi}_{i},\bm{y}_{i})+H_i(\bm{z}_i,\bm{\sigma}_{i})-J_{i}^{s}-\beta_iD_i\big)+  \sum_{n\in\mathcal{I}}\lambda_{i}(n)\big( \sum_{j\in In(i)}\sum_{f\in\mathcal{F}}x_{ji}^{f}(n)+y_{i}^{c}(n)+\sum_{f\in\mathcal{F}}y_{i}^{f}(n)-\sum_{j\in Out(i)}\sum_{f\in\mathcal{F}}x_{ij}^{f}(n)\big) \nonumber \\
	&+  \sum_{n\in\mathcal{I}}\sum_{j\in In(i)}\sum_{f\in\mathcal{F}}\tau_{ji}^{f}(n)(\xi_{ji}^{f}(n)-x_{ji}^{f}(n)) +  \sum_{n\in\mathcal{I}}\sum_{j\in Out(i)}\pi_{ij}(n)(\sigma_{ij}(n)-z_{ij}(n)) -  \rho_i\big(\sum_{n\in\mathcal{I}}\sum_{j\in In(i)}z_{ji}(n) - D_i-\gamma \nonumber \\
	&- \sum_{n\in\mathcal{I}}\sum_{j\in Out(i)}z_{ij}(n) \big) -  \mu_i\big( \sum_{k\in In(i)}\sum_{f\in\mathcal{F}}\frac{\sum_{n\in\mathcal{I}}x_{ki}^{f}(n)}{C_{ki}^{f}} + \sum_{j\in Out(i)}\sum_{f\in\mathcal{F}}\frac{\sum_{n\in\mathcal{I}}x_{ij}^{f}(n)}{C_{ij}^{f}} +\sum_{f\in\mathcal{F}}\frac{\sum_{n\in\mathcal{I}}y_{i}^{f}(n)}{C_{i}^{f}} -k_i \big)- \\
	&\sum_{j\in Out(i)}\sum_{f\in\mathcal{F}}\psi_{ij}^{f}\big(\sum_{(k,m)\in I(i,j)}\frac{ \sum_{n}x_{km}^{f}(n)}{C_{km}^{f}}+\frac{\sum_{n}x_{ij}^{f}(n)}{C_{ij}^{f}} +  \sum_{k\in \mathcal{N}^{e}(i)\cup \mathcal{N}^{e}(j)}\frac{\sum_{n}y_{k}^{f}(n)}{C_{k}^{f}} - 1 \big)\big].\nonumber
	\end{align}
	\setcounter{equation}{34}
	\hrulefill
	\vspace{-4mm}
\end{figure*}
\normalsize

Let us first give the formal NBS definition. Consider the game $\mathcal{G}=\left<\mathcal{I},\mathcal{A},\{w_i\}\right>$, where $\mathcal{I}\eq\{1,2,\dii ,I\}$ is the player set, and $\mathcal{A} \eq \mathcal{A}_1\times \mathcal{A}_2 \times\dii \times \mathcal{A}_I$ is the strategy space where $\mathcal{A}_i$ is the set of strategies (actions) available to player $i$. The payoff of each player $i$, $w_i(\cdot)$, depends on the strategy profile of all players, $\bm{a} = (a_1,a_2, \dii ,a_I)$, with $a_i \in \mathcal{A}_i$. The NBS for this game is \cite{myerson-gametheory}:

\begin{definition}[Nash Bargaining Solution--NBS]\label{def:NBS}
The strategy profile $ \bm{a}^{\ast}=(a_1^*, a_2^*,\ldots,a_I^*)$ is an NBS, if it solves the following problem:
\begin{equation}\label{eq:NBS-definition}
\begin{aligned}
\max_{\bm{a} \in \mathcal{A}} & \ \Pi_{i\in\mathcal{I}} \big(w_i(\bm{a})-w_i^d\big)
\,\,\,\mbox{s.t.}\,\, \ w_i(\bm{a}) \geq w_i^d,\ \forall\,i\in\mathcal{I}
\end{aligned}
\end{equation}
where $w_{i}^{d}$ is the disagreement point of player $i$, i.e., her payoff when an agreement is not reached.
\end{definition}

Next we consider an equivalent formulation, where the product of terms in (\ref{eq:NBS-definition}) is substituted by the sum of the corresponding logarithms \cite{mazumdar}. Hence, we define the Bargaining Optimization Problem (BOP):
\begin{align}
&\max_{\bm{x},\bm{y},\bm{z}}\sum_{i\in\mathcal{I}}\log \big(J_{i}(\bm{x}_i, \bm{x}^{in}_{i},\bm{y}_{i})+H_i(\bm{z}_i,\bm{z}^{out}_{i})-J_{i}^{s}-\beta_iD_i \big)\nonumber\\
&s.t.\,\,\,(\ref{eq:node-radio-constraint}), (\ref{eq:interference-constraint}), (\ref{eq:wifi-congestion}), (\ref{eq:cellular-congestion}), (\ref{eq:flow-conserve101})  \nonumber\\
&\sum_{n\in\mathcal{I}}\big(\sum_{j\in In(i)}z_{ji}(n)-\sum_{j\in Out(i)}z_{ij}(n)\big)\leq D_i+\gamma,\forall i\in\mathcal{I}\label{eq:vc-budget}\\
&J_{i}(\bm{x}_i, \bm{x}^{in}_{i},\bm{y}_{i})+H_i(\bm{z}_i,\bm{z}^{out}_{i})\geq J_{i}^{s}+\beta_iD_i, \forall\,i\in\mathcal{I} \label{eq:IR-constraint}\\
&0\leq x_{ij}^{f}(n)\leq C_{ij}^{f},\,\forall\,i,j,n\,\in\mathcal{I},f\in\mathcal{F}\,, \label{eq:link-capacity-add1} \\
&0\leq y_{i}^{c}(n)\leq C_{i}^{c},\,\forall\,i\in\mathcal{I},\, n\in\mathcal{I}\,, \label{eq:link-capacity-add2}\\
&0\leq y_{i}^{f}(n)\leq C_{i}^{f},\,\forall\,i,n\,\in\mathcal{I},f\in\mathcal{F}\,, \label{eq:link-capacity-add3}\\
&0\leq z_{ij}(n)\leq \sum_{i\in\mathcal{I}}(D_i+\gamma),\,\forall\,i,j,n\,\in\mathcal{I}\,, \label{eq:budget-constraint-new}
\end{align}
\noindent where the \emph{disagreement point} for each user is the sum of the standalone performance $J_{i}^{s}$ she can achieve, and the normalized virtual currency $\beta_iD_i$ she has at the beginning of the period. Eq. (\ref{eq:vc-budget}) states that users cannot have a virtual currency deficit. Eq. (\ref{eq:IR-constraint}) is the individual-rationality constraint indicating that each user will agree to cooperate only if this does not make her payoff worse. Eq. (\ref{eq:link-capacity-add1})-(\ref{eq:link-capacity-add3}) ensure that all flows will satisfy the capacity constraints of the respective links. Finally, notice that in (\ref{eq:budget-constraint-new}), each payment decision $z_{ij}(n)$ is bounded by the total virtual currency at the system. This restrains users from asking or promising payments that exceed the volume of this virtual economy (hence being infeasible) and as we will explain in the sequel, it will also facilitate the convergence of the distributed algorithm.

The BOP problem always has a non-empty feasible region. Therefore, due to (\ref{eq:IR-constraint}) there is no user of whom the payoff will decrease by participating in the service. Hence, all users are incentivized to join the service in each time period. Technically, this is ensured due to the virtual currency system and specifically the rewarding parameter $\gamma$. We should emphasize, however, that it is not necessary that all participating users will serve other users at the NBS. The optimal solution depends on the properties of the UPN graph $G$ and the users' needs and resources. The solution of the BOP problem is unique. In particular, the following lemma holds:
\begin{lemma}
The BOP problem has a unique optimal solution.
\end{lemma}
\begin{proof}
The objective function is strictly concave, since it is a composition of (strictly) concave functions. Additionally, the constraint set is compact, convex and non-empty. Notice that constraint (\ref{eq:IR-constraint}) can always be satisfied by some solution point. For example, each user $i$ can choose not to route any traffic, i.e., $x_{ij}^{f}(n)=x_{ji}^{f}(n)=0,\,\forall\,i,j,n\in\mathcal{I},\,f\in\mathcal{F}$, and only download data for herself. This way, she achieves her standalone performance, but still improves her payoff due to the participation reward $\gamma$. This also ensures that the logarithmic arguments are non-zero. Therefore, the problem has always a unique solution $(\bm{x}^{*}, \bm{y}^{*}, \bm{z}^{*})$ \cite{bertsekas-nedic}.
\end{proof}

We can derive the solution of the BOP problem by solving the KKT conditions \cite{bertsekas-nedic}. This will yield the efficient and fair Internet access and routing decisions as well as the preferable channel for each communication link. Additionally, it will determine the currency transfers among the users. Based on the system parameters, i.e., connection capacities, battery energy, and data plans, the service can offload data to Wi-Fi networks or onload data to cellular networks.

However, in all cases, the critical question is whether we can find this solution in a distributed fashion, so as to enable the distributed execution of the resource sharing mechanism. This is an important feature for such crowdsourced mobile Internet access services without central controllers.

\section{Distributed Algorithm Design for BOP}\label{sec:distributed}

The difficulties to solve the BOP problem in a distributed fashion are twofold. First, the decision variables of different users are coupled in the constraints. That is, the routing decisions of each user should take into account the capacity constraints of her neighboring nodes. Second, there is coupling in the objective functions. Namely, the logarithmic component of the BOP objective that corresponds to each user $i$ is dependent on the decision variables of her neighbors. We address these issues by introducing new auxiliary local variables for each user and \emph{consistency} constraints for each pair of neighboring users (for the coupled objectives) \cite{chiang-tutorial}. The transformed problem then has coupling only in the constraints, and can be solved using a primal-dual Lagrange iterative decomposition method \cite{chiang-tutorial}.

Let us focus on user $i\in\mathcal{I}$. Her payoff depends on her own decisions ($\bm{x}_i, \bm{y}_i, \bm{z}_i$) and the decisions $\bm{x}^{in}_{i}$, and $\bm{z}^{out}_{i}$ of her one-hop neighbors. Clearly, node $i$ can route data to her upstream neighbors $In(i)$ only if they agree on the servicing rate and the respective payments. To deal with these decision couplings, we introduce the auxiliary variables and component-wise equality constraints as follows:
\begin{align}
&\xi_{ji}^{f}(n)=x_{ji}^{f}(n),\,\forall\,i\in\mathcal{I},j\in In(i),n\in\mathcal{I},f\in\mathcal{F} \label{eq:ksi-xi}, \\
&\sigma_{ij}(n)={z}_{ij}(n),\,\forall\,i\in\mathcal{I},\,j\in Out(i),\,n\in\mathcal{I}\,. \label{eq:sigma-z}
\end{align}
We also define the matrices $\bm{\xi}_{i}=( \xi_{ji}^{f}(n)\geq 0: j\in In(i),n\in\mathcal{I},\,f\in\mathcal{F})$ and $\bm{\sigma}_{i}=( \sigma_{ij}(n):j\in Out(i),n\in\mathcal{I})$ that are maintained locally by every node $i\in\mathcal{I}$. Technically, these local variables substitute the decisions $\bm{x}^{in}_{i}$, and $\bm{z}^{out}_{i}$ of $i$'s neighbors, and hence enable her to optimize independently her payoff function. This way, each user can independently determine her downloading, routing, and payment variables, subject to coordination with her one-hop neighbors that is achieved when (\ref{eq:ksi-xi}) and (\ref{eq:sigma-z}) are satisfied.

Accordingly, we relax constraints (\ref{eq:node-radio-constraint}), (\ref{eq:interference-constraint}), (\ref{eq:flow-conserve101}), (\ref{eq:vc-budget}), (\ref{eq:ksi-xi}), and (\ref{eq:sigma-z}), and introduce the respective Lagrange multipliers:
\begin{align}
&\bm{\lambda}=(\lambda_{i}(n):i\in\mathcal{I},n\in\mathcal{I},\,n\neq i),\,\bm{\mu}=(\mu_i\geq 0,\,i\in\mathcal{I}) \nonumber\\
&\bm{\psi}=\big(\psi_{ij}^{f}\geq 0,\forall\,(i,j)\in\mathcal{E},f\in\mathcal{F} \big),\, \bm{\rho}=(\rho_i\geq 0:i\in\mathcal{I}) \nonumber\\
&\bm{\tau}=\big(\tau_{ji}^{f}(n): i,n\in\mathcal{I}, j\in In(i),\,f\in\mathcal{F}\big),\nonumber\\
&\bm{\pi}=\big(\pi_{ij}(n): i,n\in\mathcal{I}, j\in Out(i)\big). \nonumber
\end{align}

Then, we define the (partial) Lagrangian shown in eq. (34), which is separable in user-specific components $L_i(\cdot)$, $i\in\mathcal{I}$. In each iteration $q$ of the primal-dual update, the user maximizes the Lagrange function in terms of the primal variables and uses the obtained values to update the dual variables. More specifically, each user $i\in\mathcal{I}$ solves the following problem in every iteration:
\begin{align}
&\max_{\bm{x}_i, \bm{y}_i,\bm{z}_i, \bm{\xi}_i, \bm{\sigma}_i} L_i\big(\bm{x}_i, \bm{\xi}_{i},\bm{y}_{i},\bm{z}_i,\bm{\sigma}_{i} \big) \label{eq:Lagrange-user} \\
&J_{i}^{c}(\bm{x}_i, \bm{\xi}_{i},\bm{y}_{i})+H_i(\bm{z}_i,\bm{\sigma}_{i}) - J_{i}^{s}-\beta_iD_i > 0 \\
&\sum_{n\in\mathcal{I}}y_{i}^{c}(n)\leq C_{i}^{c},\,\,\sum_{f\in\mathcal{F}}\frac{\sum_{n\in\mathcal{I}}y_{i}^{f}(n)}{C_{i}^{f}}\leq 1,\, \label{eq:Lagrange-capacity-constraint2} \\
&\sum_{f\in\mathcal{F}}\frac{\sum_{n\in\mathcal{N}}x_{ij}^{f}(n)}{C_{ij}^{f}}\leq 1,\,\, \forall j\in Out(i), \\
&0\leq x_{ij}^{f}(n)\leq C_{ij}^{f},\,n\in\mathcal{I},f\in\mathcal{F},j\in Out(i), \label{eq:Lagrange-user-constraints01} \\
&0\leq \xi_{ji}^{f}(n)\leq C_{ji}^{f},\,n\in\mathcal{I},f\in\mathcal{F},j\in Out(i), \label{eq:Lagrange-user-constraints02} \\
&0\leq y_{i}^{f}(n)\leq C_{i}^{f},\,n\in\mathcal{I},\,f\in\mathcal{F}, \label{eq:Lagrange-user-constraints03} \\
&0\leq y_{i}^{c}(n)\leq C_{i}^{c},n\in\mathcal{I}, \label{eq:Lagrange-user-constraints04} \\
&0\leq z_{ji}(n)\leq K,\,\,\,n\in\mathcal{I},j\in In(i), \label{eq:Lagrange-user-constraints}
\end{align}
where the objective $L_i(\cdot)$ is:
\begin{align}
&L_i=\log\big(J_{i}(\bm{x}_i, \bm{\xi}_{i},\bm{y}_{i})+H_i(\bm{z}_i,\bm{\sigma}_{i})-J_{i}^{s}-\beta_iD_i\big) \nonumber \\
&+\sum_{n\in\mathcal{I}}\big[\lambda_{i}(n)(y_{i}^{c}(n)+\sum_{f\in\mathcal{F}}y_{i}^{f}(n))\nonumber \\
&-\sum_{j\in Out(i)}\sum_{f\in\mathcal{F}}x_{ij}^{f}(n)(\lambda_{i}{(n)}-\lambda_{j}{(n)}) \big]-\rho_i\sum_{n\in\mathcal{I}}\sum_{j\in In(i)} z_{ji}(n) \nonumber \\
& + \sum_{j\in In(i)}\big(\rho_j\sum_{n\in\mathcal{I}}z_{ji}(n) + \sum_{n\in\mathcal{I}}\sum_{f\in\mathcal{F}}\tau_{ji}^{f}(n)\xi_{ji}^{f}(n)\big) \nonumber \\
&-\sum_{n\in\mathcal{I}} \big[ \sum_{j\in Out(i)}\big( \sum_{f\in\mathcal{F}}x_{ij}^{f}(n)\tau_{ij}^{f}(n) - \pi_{ij}(n)\sigma_{ij}(n)\big) \nonumber \\
&+\sum_{j\in In(i)}\pi_{ji}(n)z_{ji}(n)\big]-\sum_{j\in Out(i)}\mu_j\sum_{f\in\mathcal{F}}\frac{\sum_{n\in\mathcal{I}}x_{ij}^{f}(n)}{C_{ij}^{f}} \nonumber \\
&-\mu_i\big[ \sum_{j\in Out(i)}\sum_{f\in\mathcal{F}}\frac{\sum_{n\in\mathcal{I}}x_{ij}^{f}(n)}{C_{ij}^{f}} + \sum_{f\in\mathcal{F}}\frac{\sum_{n\in\mathcal{I}}y_{i}^{f}(n)}{C_{i}^{f}} \big] \nonumber \\
&-\sum_{j\in Out(i)}\sum_{f\in\mathcal{F}}\psi_{ij}^{f}\big[ \frac{\sum_{n\in\mathcal{I}}x_{ij}^{f}(n)}{C_{ij}^{f}} +\frac{\sum_{n}y_{i}^{f}(n)}{C_{i}^{f}} \big]\nonumber \\
&-\sum_{k\in \mathcal{N}^{e}(i)}\sum_{f\in\mathcal{F}}\psi_{ki}^{f}\big[\frac{\sum_{n}y_{i}^{f}(n)}{C_{i}^{f}}+ \sum_{m\in Out(i)}\frac{\sum_{n\in\mathcal{I}}x_{im}^{f}(n)}{C_{im}^{f}}\big].\nonumber
\end{align}
The solution yields the optimal (in the current iteration $q$) values $\bm{x}_{i}^{(q)}, \bm{y}_{i}^{(q)}, \bm{z}_{i}^{(q)}, \bm{\xi}_{i}^{(q)}, \bm{\sigma}_{i}^{(q)}$.

User $i$ then uses the above values of the primal variables to calculate the gradients and update the dual variables \cite{bertsekas-nedic}:
\begin{align}
\lambda_{i}^{(q+1)}(n)&=\lambda_{i}^{(q)}(n)+s^{(q)} \big[\sum_{j\in In(i)}\sum_{f\in\mathcal{F}}x_{ji}^{f\,(q)}(n)+y_{i}^{c\,(q)}(n)\nonumber \\
&+\sum_{f\in\mathcal{F}}y_{i}^{f\,(q)}(n)-\sum_{j\in Out(i)}\sum_{f\in\mathcal{F}}x_{ij}^{f\,(q)}(n)\big]\label{eq:lambda-update}
\end{align}
\begin{equation}
\tau_{ji}^{f\,(q+1)}(n)=\tau_{ji}^{f\,(q)}(n)+s^{(q)}\cdot\big(\xi_{ji}^{f\,(q)}(n)-x_{ji}^{f\,(q)}(n) \big) \label{eq:tau-update}
\end{equation}
\begin{equation}
\pi_{ij}^{(q+1)}(n)=\pi_{ij}^{(q)}(n)+s^{(q)}\cdot\big(\sigma_{ij}^{(q)}(n)-z_{ij}^{(q)}(n) \big) \label{eq:pi-update}
\end{equation}
\begin{eqnarray}
\rho_{i}^{(q+1)}&=&\big[\rho_{i}^{(q)}+s^{(q)}\big(\sum_{n\in\mathcal{I}}(\sum_{j\in In(i)}z_{ji}^{(q)}(n) \nonumber \\
&-&\sum_{j\in Out(i)}z_{ij}^{(q)}(n)) - \gamma - D_i \big)\big]^{+} \label{eq:rho-update}
\end{eqnarray}
where $[\cdot]^{+}$ denotes the projection onto the non-negative orthant and $s^{(q)}\geq 0$ is a properly selected step during iteration $q$ \cite{bertsekas-nedic}. A similar formula is also applied for updating $\mu_i\geq 0$ and $\psi_{ij}^{f}\geq 0$, based on the set of equations (\ref{eq:node-radio-constraint}) and (\ref{eq:interference-constraint}), respectively. Finally, each user passes the updated dual variables to her one-hop neighbors, who will use them to optimize the primal variables in the next iteration.

\begin{algorithm}[t]
\small
\SetKwInOut{Input}{input}\SetKwInOut{Output}{output}
\Output{$\bm{x}^{*}$, $\bm{y}^{*}$, $\bm{z}^{*}$}
\BlankLine
\nl $q \leftarrow 0$;\\%
\nl Set $\bm{x}^{(0)}$, $\bm{y}^{(0)}$, $\bm{z}^{(0)}$, $\bm{\xi}^{(0)}$, $\bm{\sigma}^{(0)}$, $\bm{\lambda}^{(0)}$, $\bm{\tau}^{(0)}$, $\bm{\pi}^{(0)}$, $\bm{\rho}^{(0)}$, $\bm{\mu}^{(0)}$, $\bm{\psi}^{(0)}$, $\epsilon$;\\%
\nl $\mathrm{conv\_flag}$ $\leftarrow$ 0; \# \emph{initialize the convergence flag} \\%
\nl \While{$\mathrm{conv\_flag}=0$}{
\nl $q \leftarrow q+1$; \\%
\nl \For {$i=1:I$}{
\nl Solve (\ref{eq:Lagrange-user}- \ref{eq:Lagrange-user-constraints}) for primal vars $\bm{x}_{i}^{(q)}, \bm{y}_{i}^{(q)}, \bm{z}_{i}^{(q)}, \bm{\xi}_{i}^{(q)}, \bm{\sigma}_{i}^{(q)}$\\%
\nl Send $x_{ij}^{f\,(q)}(n),\forall n\in\mathcal{I}\setminus\{i\}$, $f\in\mathcal{F}$, to $j\in Out(i)$;\\%
\nl Send $z_{ji}^{(q)}(n),\forall n\in\mathcal{I}\setminus\{i\}$, to $j\in In(i)$;\\%
\nl \For {$j=1:I$, $n=1:I$}{
\nl Calculate dual vars $\lambda_{i}^{(q+1)}(n), \tau_{ji}^{f\,(q+1)}(n), \pi_{ij}^{(q+1)}(n),$
$\rho_{i}^{(q+1)}, \psi_{ij}^{f\,(q+1)}, \mu_{i}^{(q+1)}$ using (\ref{eq:lambda-update}-\ref{eq:rho-update}); \\}%
\nl Send $\lambda_{i}^{(q+1)}(n),\, \tau_{ji}^{f\,(q+1)}(n),\,\mu_{i}^{(q+1)} \forall n\in\mathcal{I}\setminus\{i\}$, $f\in\mathcal{F}$, to $j\in In(i)$;\\%
\nl Send $\rho_{i}^{(q+1)},\,\pi_{ij}^{(q+1)}(n),\,\psi_{ij}^{f\,(q+1)} \forall n\in\mathcal{I}\setminus\{i\},\,f\in\mathcal{F}$, to $j\in Out(i)$;\\}%
\nl \If {$|[\lambda_{i}^{(q+1)}(n)-\lambda_{i}^{(q)}(n)]/\lambda_{i}^{(q)}(n)|<\epsilon$ \emph{and} $|[\rho_{i}^{(q+1)}-\rho_{i}^{(q)}]/\rho_{i}^{(q)}|<\epsilon$ \emph{and} $|[\pi_{ij}^{(q+1)}(n)-\pi_{ij}^{(q)}(n)]/\pi_{ij}^{(q)}(n)|<\epsilon$ \emph{and} $|[\tau_{ji}^{f\,(q+1)}(n)-\tau_{ji}^{f\,(q)}(n)]/\tau_{ji}^{f\,(q)}(n)|<\epsilon$ \emph{and} $|[\psi_{ij}^{f\,(q+1)}-\psi_{ij}^{f\,(q)}]/\psi_{ij}^{f\,(q)}|<\epsilon$ $|[\mu_{i}^{(q+1)}-\mu_{i}^{(q)}]/\mu_{i}^{(q)}|<\epsilon$
$\forall i, n \in\mathcal{I}, f\in\mathcal{F}, j\in Out(i)$}{$\mathrm{conv\_flag}\leftarrow 1$;}}%
\caption{Distributed Solution of BOP}\label{Distributed_Algo}
\end{algorithm}
\normalsize



The Algorithm is executed in a synchronous fashion, which requires a common clock of all users and a small delay for message passing (circulation of the dual and primal variables). This is a reasonable assumption for small-scale crowdsourced connectivity networks in a small neighborhood. The complete algorithm is summarized in Algorithm \ref{Distributed_Algo} and is provably converging to the optimal solution.
\begin{lemma}
Algorithm \ref{Distributed_Algo} globally converges to the optimal solution ($\bm{x}^{*}, \bm{y}^{*}, \bm{z}^{*}$) of the BOP problem, under properly chosen step sizes $s^{(q)}$ for each iteration $q$.
\end{lemma}
\begin{proof}
BOP has a strictly concave objective and a closed, non-empty and convex constraint set. Thus, Algorithm $1$ converges to optimal solution \cite{bertsekas-nedic} if (i) the step sequence $s^{(q)}$, $q=1,2,\ldots,$ is properly selected, and (ii) the gradients used in (\ref{eq:lambda-update})-(\ref{eq:rho-update}), are bounded. Consider user $i$, and we see that the variables $x_{ij}^{f}(n)$, $y_{i}^{f}(n)$, $y_{i}^{c}(n)$, $\xi_{ji}^{f}(n)$ are upper bounded by constraints (\ref{eq:Lagrange-user-constraints01}) - (\ref{eq:Lagrange-user-constraints04}) and the energy cost function (since it is strictly convex), and $z_{ji}(n), n\in\mathcal{I}, j\in Out(i)$ are positive and upper bounded by $K$. Hence, if we employ a diminishing step size, e.g. $s^{(q)}=(1+m)/(t+m)$ with $m\geq 0$, then the convergence is guaranteed \cite{chiang-tutorial}, \cite{bertsekas-nedic}.
\end{proof}

Notice that each user passes messages for each commodity $n\in\mathcal{I}$, only to her one-hop neighbors. Hence, the message passing overhead of the algorithm is $O(\bar{d}|\mathcal{I}||\mathcal{F}|)$, where $\bar{d}$ is the average degree of the graph $G$. However, we expect that the number of users in the group will be small (due to the need to be in proximity), hence even a complexity of $O(|\mathcal{F}||\mathcal{I}|^2)$ (assuming a fully connected graph) is affordable. \rev{Finally, note that this required exchange of messages will induce energy cost to the devices, additional to the energy expenditure due to the data downloading and relaying. The actual energy impact of this coordination depends on the technology used (e.g., how often the UPN is reconfigured, as we experimentally showed in \cite{syrivelis}) and it's study is beyond the scope of this work.}

\section{Performance Evaluation}\label{sec:numerical}

In this section we employ a representative system setup, and demonstrate how the collaborative connectivity service performs in a variety of different scenarios. The system parameters follow experimental studies \cite{huang-mobisys}, \cite{balasubra-imc}, \cite{tan-3g-evaluation}, \cite{perrucci-vtc}.


\textbf{Simulation Setup}. We consider a set of $|\mathcal{I}|=6$ users randomly placed in a geographic area, and study their interactions for a time period of $T=100$ seconds. The Internet access capacity of each user depends on whether she uses a cellular LTE-A, 3G, or a Wi-Fi connection. Field experiments have measured the actual average speed to be $12.74$ Mbps for LTE-A, $4.12$ Mbps for Wi-Fi, and $1$ Mbps for 3G networks \cite{huang-mobisys}, \cite{tan-3g-evaluation}. Moreover, we assume that there are $3$ orthogonal channels for Wi-Fi which, in general, may have different capacities due to the surrounding (background) interference beyond the UPN transmissions.

The users communicate with each other using Wi-Fi Direct, and each user $i$ has one NIC, i.e., $k_i=1$. The achievable rate among two users $i$ and $j$ decreases with their distance $d_{ij}$ (in meters). In order to account for a representative setting with random channel conditions, we assume that two users separated by $1$ meter can achieve a maximum communication speed of $64$Mbps; the speed drops to $0.1$ Mbps when the distance increases to $30$ meters. The rate will be smaller if there is interference, and zero when the distance is larger than $30$ meters. Therefore, the average data rate $C_{ij}^{f}$ that can be transferred over each link $(i,j)\in\mathcal{E}$ and channel $f$ satisfies:
\begin{equation}
C_{ij}^{f} = b_{ij}^{f}\cdot100\log(1+ 0.9/d_{ij}^2)\,.
\end{equation}
where $b_{ij}^{f}\in [0.5,1]$ is a uniformly random variable modeling the effect of surrounding interference in channel $f\in\{1,2,3\}$.

\rev{
The energy consumed by a data transfer is proportional to the data volume, the transmission power and rate, and the channel conditions (e.g., due to distance and packet retransmissions \cite{perrucci-vtc}) \cite{balasubra-imc}. We use here eq. (\ref{eq:energy-equation001}), where the parameters are selected according to \cite{huang-mobisys}: the energy consumption for 3G links is twice as for LTE links, and 4 times higher than Wi-Fi. Also, uplink transmission consumes 8 times more energy in LTE and 3G, and 2 times more in Wi-Fi, than downlink transmissions. Note that for our comparative study the above relative values are adequate enough.}

Every user $i\in\mathcal{I}$ has a logarithmic utility function
\begin{equation}
U_i=\alpha_{i}\log\big(1+r_i\big),
\end{equation}
which satisfies the principle of diminishing marginal returns, and $r_i$ is given by eq. (\ref{eq:r-i-function}). Parameter $\alpha_i\in[0,1]$ captures the different communication needs of the different users. Also, the virtual currency parameters $\beta_i$, $\forall i\in\mathcal{I}$, are uniformly distributed in $(0,1]$. \rev{Finally, regarding the cellular data costs, we employ two approaches to make our study more comprehensive. First, we assume that the user's dissatisfaction can be a simple linear function of consumed data amount, and we use the representative price reported by ITU (for UK) \cite{itu-report}, i.e., $0.002\$$/Mbit. Second, we also use the psychological price given by eq. (\ref{eq:quota}), where the quotas are selected randomly from the interval $[0,5]$GBytes, unless otherwise specified. Finally, for both cases, we set $p_i=0$ for users who have an unlimited cellular data plan or using Wi-Fi connections.}

\begin{figure}[t]%
\centering
\subfigure
{\epsfig{figure=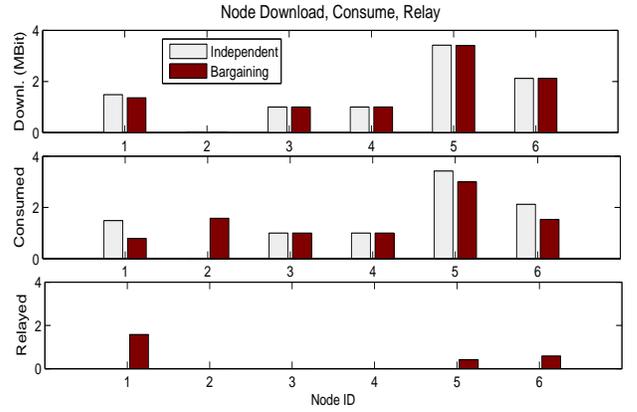, width=8.9cm,height=5.6cm}}\qquad%
\vspace{-2mm}
\caption{Comparison of independent, and bargained solution. System parameters: $\{C_{i}\}_{i}=\{9.7, 0.0, 1.0, 1.0, 4.12, 2.1\}$Mbps, $\{p_i\}_{i}=\{0.02, 0.02, 0.008, 0.001, 0.0, 0.0\}$ \$/Mbit. User $2$ has higher demand than other users, i.e., $\alpha_{i}=2,\,i\in\{1,3,4,5,6\}$ and $\alpha_2=4$. The figure shows the total amount of downloaded, consumed, and relayed data per second.}
\label{fig:ssb}
\vspace{-2mm}
\end{figure}

\begin{figure} 
	\vspace{-1mm}
		\subfigure{\scriptsize \mytab}%
		\subfigure{
		\epsfig{figure=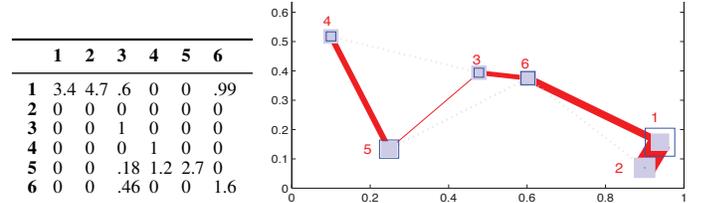,width=5.6cm,height=2.7cm}}
	\rev{\caption{Detailed example of the UPN operation under the Nash bargaining solution. Users' parameters are given in caption of Fig. \ref{fig:ssb}. In the graph on the right, blue squares represent amounts of downloaded data by each node, and light-blue shaded boxes the amounts of consumed data. Dotted lines show connectivity and red lines represent the amount of relayed data. Actual values are shown on the left table.}}
		\label{fig:UPN-snapshot}
	\vspace{-1mm}
\end{figure}


\subsubsection{Comparison of Bargained and Independent Solutions}
\rev{Our first goal is to study the UPN operation from the point-of-view of the user. That is, we wish to demonstrate the impact of collaboration on the amount of data that each user consumes, and on the users' payoff.} We compare the ``bargained'' payoff $J_{i}^{G}$ for each user $i$ with the ``independent'' payoff $J_{i}^{s}$. We consider a setting with $6$ users, where user $1$ has an LTE connection, user $2$ does not have Internet access, users $3$ and $4$ have 3G connections, and users $5$ and $6$ have Wi-Fi connections. The Internet access capacity and price values (linear cost is assumed here) are shown in the caption of Figure \ref{fig:ssb}. Also, user $i=2$ has higher utility ($\alpha_2=4$) compared to the other users ($\alpha_i=2,\,i\neq 2$). For simplicity, all other system parameters have been set equal for all users. The results represent the average obtained over $50$ experiments for different user locations, namely uniformly distributed in the $[0,100m]\times [0,100m]$ plane.

In Figure \ref{fig:ssb} we plot the total amount of data each user downloads, consumes, and relays for both scenarios (participation in UPN or not). We observe that users download (almost) the same amount of data in both cases, apart from user $2$ who cannot access the Internet independently. In the bargained scenario, users consume different amounts of data since some of them relay traffic for others. For example, user $1$ consumes on average $46\%$ less data in the bargained scenario compared to her standalone operation (she routes the rest $54\%$ to her neighbors). \rev{This cooperation depends heavily on user $1$'s valuation for the virtual currency, and thus for receiving UPN service in the future}. It is clear that the extent to which the users will relay traffic for others depends on the relative value of the virtual currency to their aggregate opportunity cost (not downloading for their own needs), Internet usage, and energy consumption costs. \rev{To further help the reader build intuition, we present a detailed snapshot of the UPN operation for one of the 50 experiments in Figure 5.}


\begin{figure}[t]
\vspace{-1mm}
\begin{center}
\subfigure[]{
\epsfig{figure=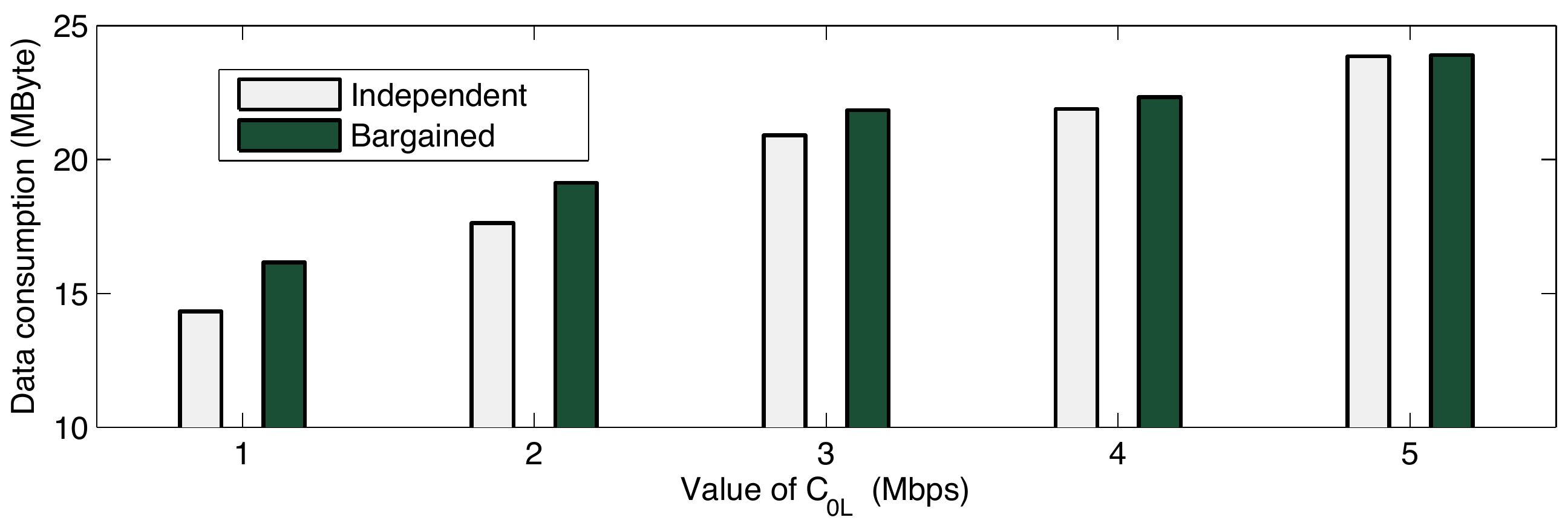, width=8.4cm,height=3.3cm}}
\subfigure[]{
\epsfig{figure=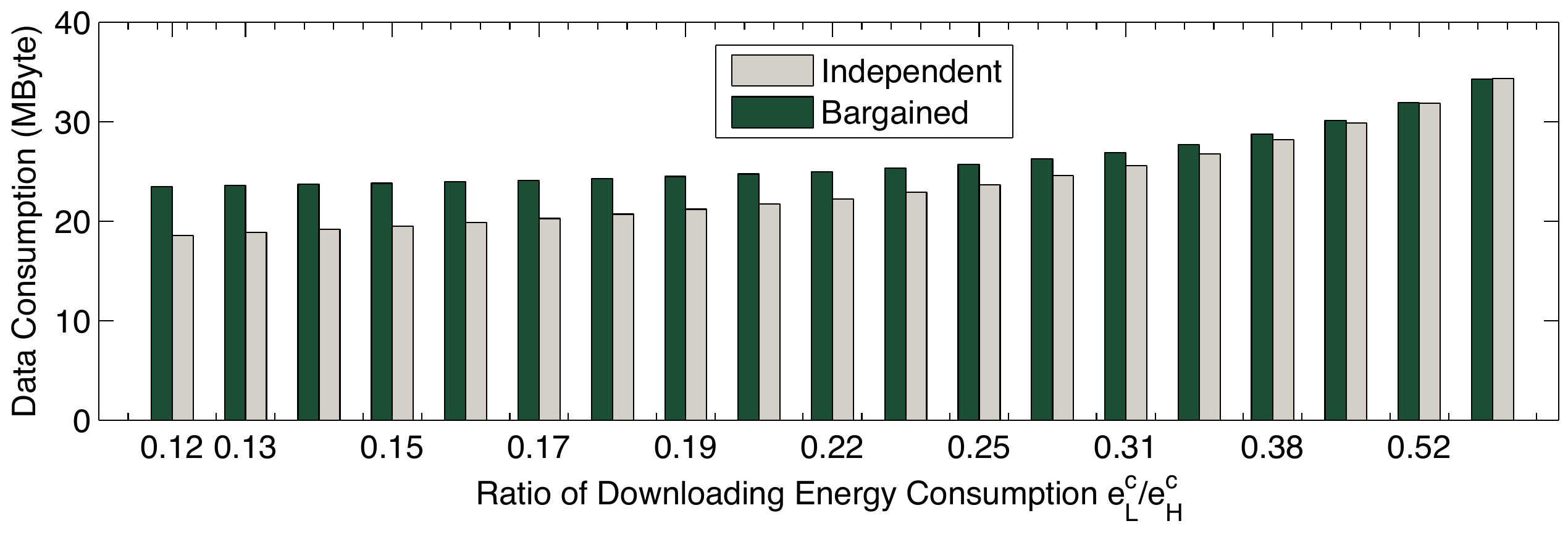, width=8.4cm,height=3.3cm}}
\subfigure[]{
\epsfig{figure=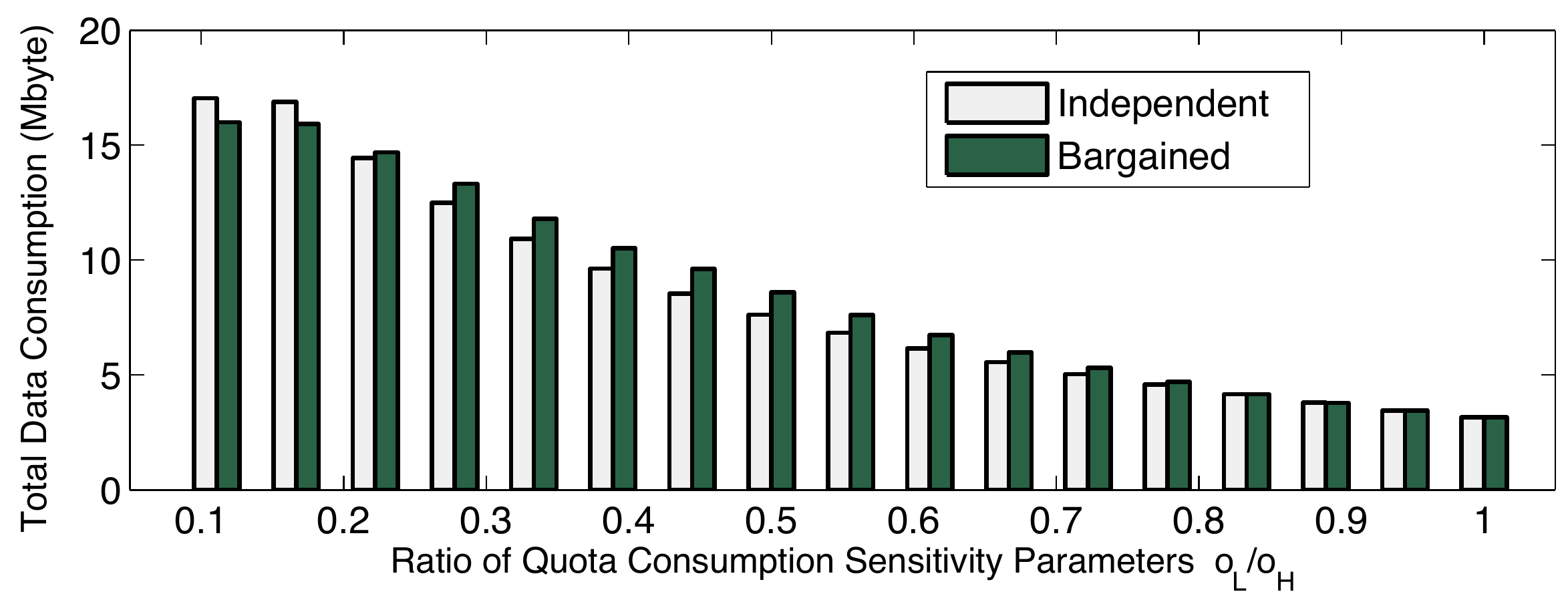, width=8.4cm,height=3.3cm}}
\caption{\rev{Impact of user diversity on collaboration benefits. All results are averaged over $30$ runs with uniformly distributed user locations in $[0,100m]\times [0,100m]$. Users are identical in all aspects, except the following parameters \textbf{(a)}: Diverse Internet access: $C_{01}=C_{02}=C_{0H}=12.7$ Mbps, $C_{03}=C_{04}=C_{05}=C_{06}=C_{0L}$ Mbps; \textbf{(b)}: Diverse energy consumption per downloaded byte: $e_{1}^{c}=e_{2}^{c}=e_{L}^{c}$, and the remaining four have the same higher value $e_{H}^{c}$; \textbf{(c)}: Diverse data plan parameters: $o_1=o_{2}=o_{L}$, and the other four have the same higher value $o_{H}$.}}
\label{fig:sims-comparision}
\end{center}
\vspace{-1mm}
\end{figure}

\subsubsection{Impact of Users Diversity on UPN Performance} \rev{Next we investigate how the performance benefits of the UPN service depend on the users' diversity, specifically on the differences in their (i) Internet access capacity, (ii) energy consumption parameters, and (iii) data plan parameters.} Starting with case (i), we assume that two users have the same high Internet access capacity ($C_{0H}=12.7$ Mbps) and the other four identical low Internet capacity $C_{0L}$ that gradually increases from $1$ Mbps to $5$ Mbps. All other parameters are fixed and equal for all six users. In Fig. \ref{fig:sims-comparision}(a) we plot the aggregate amount of downloaded data for the bargained and independent solutions. \rev{We observe that as users become less diverse (value of $C_{0L}$ increases), the gap of total downloaded data (per time period) for the bargained solution compared to the standalone solution decreases from $29.36\%$ to almost $0\%$. We can observe a similar trend regarding the user diversity in terms of energy consumption. Namely, we consider the above setup and assume that two users have low energy consumption cost $e_{L}^{c}=0.15$ J/Mbit and the other four identical high energy consumption parameter $e_{H}^{c}=1.13$ J/Mbit which decreases gradually to $0.24$. In this case, the benefit of the service decreases from $15\%$ to almost $0\%$, as is shown in Fig. \ref{fig:sims-comparision}(b)}. The intuition is that the more diverse the Internet capacities or the energy consumption parameters of the users are, the larger the benefits from their cooperation.

\rev{Finally, we study the impact of data plan diversity in Figure \ref{fig:sims-comparision}(c). We assume that the cellular access cost is given by function $Q_i(y_{i}^{ag,c})$ (eq. (\ref{eq:quota})), and that the only difference of users is in parameter $o_i$ (similar results obtained when the diversity is in $A_i$). That is, two of them have lower value ($o_L$) than the other four users ($o_H$), and this ratio increases along the x-axis, becoming eventually $o_L/o_H=1$ (no diversity case). As it is expected, the cooperation benefit diminishes as the diversity among the psychological prices of the users decreases. Interestingly, however, for very diverse scenarios ($o_L/o_H<0.2$), the nodes download more data under the independent operation than the bargained operation. In this case, the four users with the expensive data plans find it much more beneficial to buy data (using the virtual currency) from the other two users in the bargained scenario. Correspondingly, the two users with the low-cost data plans assess the trade off between consuming data and downloading (and relaying) for their neighbors, and find the latter to be more beneficial (due to the virtual currency payments). Such a decision reduces the total data downloaded by these two users, due to the additional energy consumption when relaying data for others. This simple example reveals that the parameters of the users can have complicated impacts on the final network operation. Our algorithm ensures that the UPN operation is optimal as it maximizes the users' payoffs according to the NBS.}

\begin{figure}[t]%
\centering
\subfigure
{\epsfig{figure=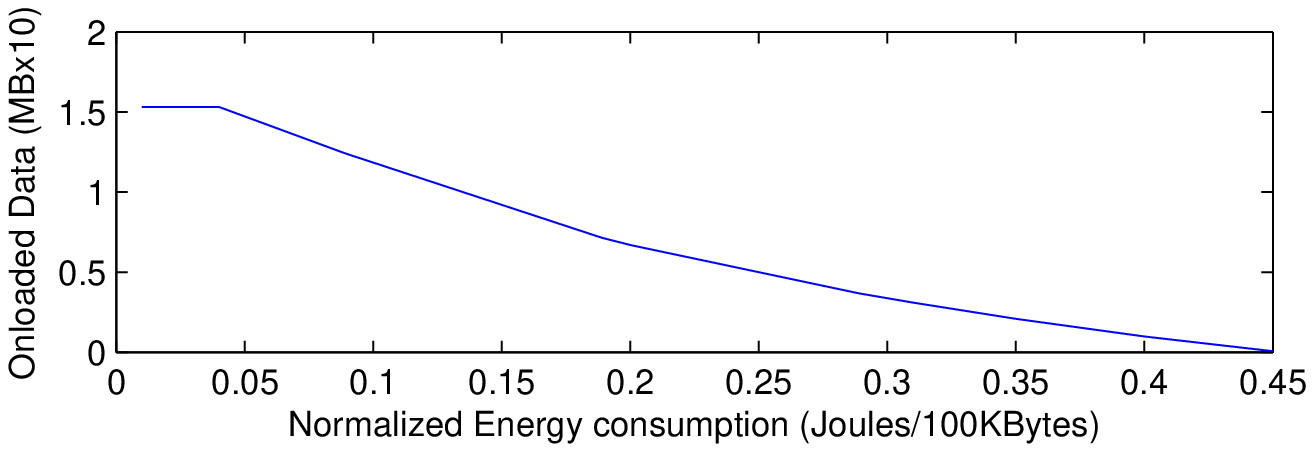,
width=7.5cm,height=2.55cm}}\qquad%
\subfigure 
{\epsfig{figure=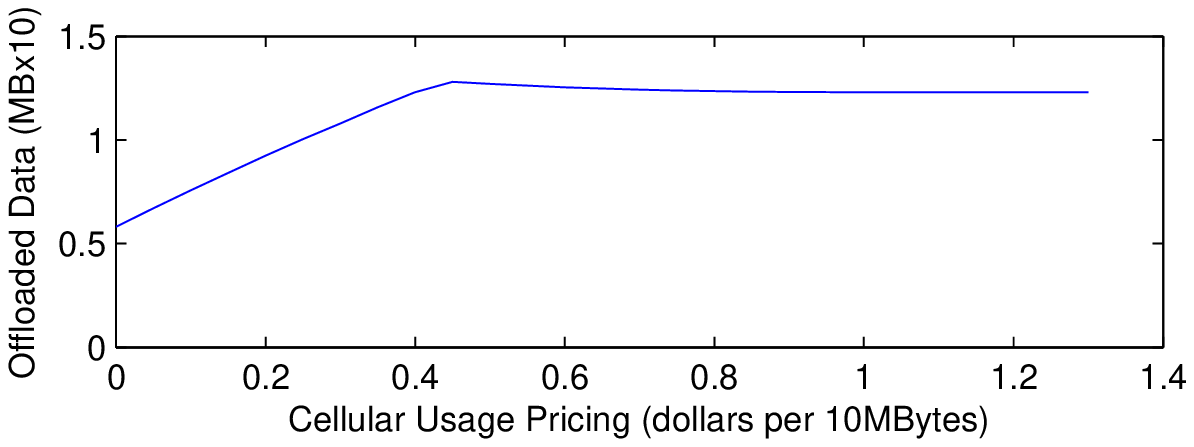,width=7.5cm,height=2.55cm}} %
\caption{Upper subfigure (Onloading): LTE-A capacity ($10$ Mbps) is $5$ times larger than the Wi-Fi capacity ($2$ Mbps), mobile data price is $p=0.002$ \$/Mbit, and the LTE link energy consumption 3-times larger than the Wi-Fi link energy consumption. Lower subfigure (Offloading): LTE capacity ($4$ Mbps) is twice the Wi-Fi capacity, and data usage price increases.}
\label{fig:offload-onload}
\end{figure}

\subsubsection{Offloading and Onloading Capabilities of UPNs}

\rev{Finally, we consider the UPN operation from the networks's point of view, and explore how the traffic is shifted among different networks due to cooperation.} Namely, the cooperating users may route data that was intended for a cellular network to a Wi-Fi network (\emph{offloading}), or the other way around (\emph{onloading}). The latter option becomes attractive when Wi-Fi links are congested and the cellular access capacity of some users is high and of low cost. For this specific experiment, we employ the setup shown in Figure \ref{fig:offloading-onloading} where two users, one with Wi-Fi Internet access (\emph{user 1}) and the other with a cellular access (\emph{user 2}), form a UPN.

First we study onloading (Figure \ref{fig:offloading-onloading} right). We assume that the cellular capacity (10 Mbps) is $5$ times larger than the Wi-Fi capacity (2 Mbps). The rest of the parameters are identical for the two users. The mobile data price is equal to $p=0.002$ \$/Mbit, and the LTE link energy consumption 3-times larger than the Wi-Fi link energy consumption. In this setting we aim to investigate how the onloading is affected by the energy consumption that the gateway node incurs. In the upper subfigure of Figure \ref{fig:offload-onload} we depict the amount of onloaded data, i.e., the data that user 2 (gateway) downloads and delivers for user 1 (client), as a function of the relaying energy consumption ($e_{21}^{f,s}$). The latter can vary if, for example, the distance among the two nodes changes. We observe that for a certain value range of $e_{21}^{f,s}$ the amount of onloaded data is maximum and constant. However, as $e_{21}^{f,s}$ increases the BOP solution yields smaller amounts of relayed data. This result verifies again that the BOP solution takes into account all the users' characteristics and balances in a fair fashion the costs and benefits of the UPN nodes.

\rev{Next we explore the offloading case (Figure \ref{fig:offloading-onloading} left).} We assume the cellular capacity is only twice of the Wi-Fi one (4 Mbps versus 2 Mbps). Having all the other parameters fixed (and equal for the two nodes), we investigate how the amount of data that user 1 (gateway) offloads for user 2 (client) varies as the Internet access price increases. As we observe in the lower subfigure in Figure \ref{fig:offload-onload}, the amount of offloaded data initially increases with the price per byte paid by the LTE user $2$. This is expected since, as the mobile Internet becomes more expensive, user $2$ prefers to pay the Wi-Fi user with virtual currency for the offloading, instead of downloading the content from her cellular link. More interestingly, after a certain point (when cellular pricing reaches $0.44$), the offloaded data decreases (compared to its maximum value) and then remains constant despite the further cellular price increase. This interesting result is due to the employed fairness criterion which determines the performance that each user will receive in the UPN in analogy to her standalone performance. Therefore, since here the standalone performance $J_{2}^{s}$ of user $2$ is reduced (i.e., due to the increasing prices user $2$ would download less data in standalone operation), the respective performance $J_{2}$ within the UPN (and hence the amount of offloaded data) first decreases and accordingly remains constant.

\section{Related Works} \label{sec:RelatedWork}


One of the first UPN examples is the Wi-Fi community networks \cite{fon}. The key challenges there include security issues \cite{crowcroft} and user participation incentives \cite{efstathiou-wifi}. Similar mechanisms have been studied for ad hoc networks \cite{crowcroft-peva}, and wireless mesh networks \cite{mandayam-cooperative}. These results are not directly applicable to mobile UPNs, since they do not account for users' different types of resources nor for their data usage cost. Also, most UPN users can access the Internet without relying on others' help, while this is typically not the case for other autonomous networks. The characterization of such standalone operations is critical in determining whether a user will agree to join the service or not. This is the reason we employ the Nash bargaining solution \cite{nash-bargaining}, \cite{mazumdar}.

UPN services can be centralized (as in FON), or decentralized where users negotiate with each other. For the former, \cite{guerin-infocom} studied a pricing rule for inducing service adoption, and \cite{manshaei} analyzed the price competition among FON-like operators and conventional operators. For decentralized services, \cite{lui-mesh} and \cite{walrand-wifi} performed game-theoretic analysis to predict the prices users charge to each other. \rev{Our scheme differs in that each user can have many roles and that multiple users can concurrently collaborate. Furthermore, these prior studies did not account for the limited data quotas and the energy limitations of users. On the contrary, we propose a detailed approach for studying the impact of quota dynamics \cite{andrews-journal} on users' collaboration. We follow the same approach to qualitatively study the effects of energy consumption, since it is known that the analytical relation of energy consumption and mobile transmissions (especially in the unlicensed band) is essentially intractable, e.g., \cite{nilsson}, \cite{koutsonikolas-infocom15}, \cite{energy-serrano}.} Finally, another unique aspect of our model is that the interactions among users from heterogeneous networks (cellular and Wi-Fi), apart from offloading, enable also the onloading of Wi-Fi traffic to cellular networks \cite{laoutaris-onloading} as illustrated in Figure \ref{fig:offloading-onloading}. Such flexible cooperation framework differs from previous offloading-only architectures \cite{iosifidis-wiopt}.

In the context of mobile cooperative networks, references \cite{cool-tether} and \cite{combine-mobisys07} proposed energy-prudent architectures that aggregate the cellular bandwidth of multiple hosts to build mobile hotspots. Similarly, \cite{crowdmac} presented a scheduling scheme that allows hosts to dynamically admit client requests so as to maximize their revenue. When each user can serve both as a client and host, a decision framework should assign these roles to users, based on their residual battery energy \cite{jung-mcs10}. A similar model was proposed in \cite{sharma-aggr}. Many related works in this area assume that users have strong social ties \cite{fragouli} or they are interested in the same content \cite{markopoulou}, hence there is no need for incentive provision schemes. This assumption was relaxed in \cite{indapson-infocom14} that proposed a rewarding scheme for one-hop UPNs. \rev{Moreover, the above works did not consider heterogeneous Internet access links nor multihop data delivery solutions, as they mainly studied the collaboration among two users over single-hop connections}. Such architectures are now implementable due to software defined networking (SDN) as we demonstrated recently in \cite{syrivelis}.

The proposed Internet access and routing decisions can be supported either by distributed scheduling and routing algorithms that were proposed for ad hoc and mesh networks \cite{kodialam}, \cite{XLin}, if we assume that the devices can operate in a perfectly synchronous fashion (e.g., employing a TDM MAC protocol \cite{tdm-mac}), or they can be implemented by the typical CSMA/CA protocol which does not require tight synchronization. However, there might be a gap between the theoretical performance studied in this paper and the practical implementation of such schemes, which we intent to analyze in our future works.




\section{Conclusions and Discussion}\label{sec:conclusions}

User-provided networks take advantage of the technical capabilities of handheld user-owned devices; and connect different and possibly heterogeneous networks in a bottom up fashion. This constitutes a paradigm shift with the potential to increase the effective capacity of wireless networks by unleashing dormant network resources. One of the main challenges in UPNs is to incentivize the participation of users, on the basis of a fair resource contribution and capacity allocation. In this work, we proposed an optimization framework which maximizes the UPN efficiency, and allocates the produced capacity in a fair fashion. The mechanism is amenable to distributed implementation, and is lightweight in terms of communication overheads. We verified numerically that the system ensures a higher performance than the independent operation of the nodes, and that the UPN benefits increase when users have more diverse needs and/or resources.

Interestingly, the proposed framework constitutes a method for unifying energy and monetary costs, based on the users' needs and demands, and enables the monetization of the technical capabilities of user devices. Unlike other cooperative mechanisms that rely on auctions or centralized pricing schemes, this bargaining-based solution is self-enforcing and can be implemented in a decentralized fashion. Furthermore, the proposed framework involves interactions only between one-hop neighbors. Besides, it was made clear in the analysis that UPNs extend the offloading architectures by enabling onloading traffic from Wi-Fi to cellular networks, hence provide a flexible solution for addressing the increasing interference and collision problem in the ISM band.

Finally, it worths mentioning that our approach is motivated by the concept of collaborative consumption, which promotes business models and systems for sharing resources. The term was introduced in $1978$ by Felson \cite{CC-felson} and has been recently revisited in a comprehensive fashion \cite{CC-RachelBook}. Interestingly, several innovative startups \cite{opengarden}, \cite{karma}, \cite{m-87}, have recently introduced respective services for mobile users. Besides, recent measurement studies have reveled the potential gains of mobile user collaboration in cellular networks \cite{laoutaris-conext14}, which can nowadays be implemented due to recent technological advancements in wireless networking \cite{qualcom}, \cite{keown}, \cite{syrivelis}, \cite{cisco-2012_1}.

\section{Acknowledgments}\label{sec:acks}

This work is supported by the General Research Funds (Project No. CUHK 14202814 and 14206315) established under the University Grant Committee of the Hong Kong Special Administrative Region, China; and by the National Science Foundation, US, under Grant CNS 1527090. 


%
%


\ifCLASSOPTIONcaptionsoff
  \newpage
\fi

\begin{IEEEbiography}[{\includegraphics[width=1in,height=1.45in,clip,keepaspectratio]{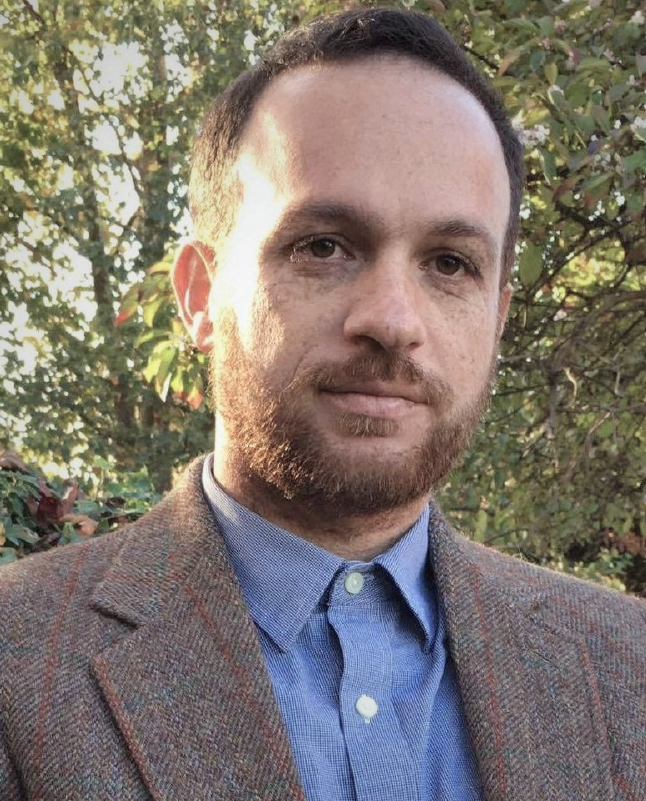}}] {George Iosifidis} received the Diploma degree in electronics and telecommunications engineering from the Greek Air Force Academy in 2000, and the M.S. and Ph.D. degrees in electrical engineering from the University of Thessaly, Greece, in 2007 and 2012, respectively. He was a Post-doctoral researcher with CERTH, Greece, and Yale University, USA. He is currently the Ussher Assistant Professor in Future Networks with the Trinity College Dublin, University of Dublin, and also a Funded Investigator with the telecommunications research centre CONNECT, Ireland. His research interests lie in the broad area of wireless network optimization and network economics.
\end{IEEEbiography}

\begin{IEEEbiography}[{\includegraphics[width=1in,height=1.25in,clip,keepaspectratio]{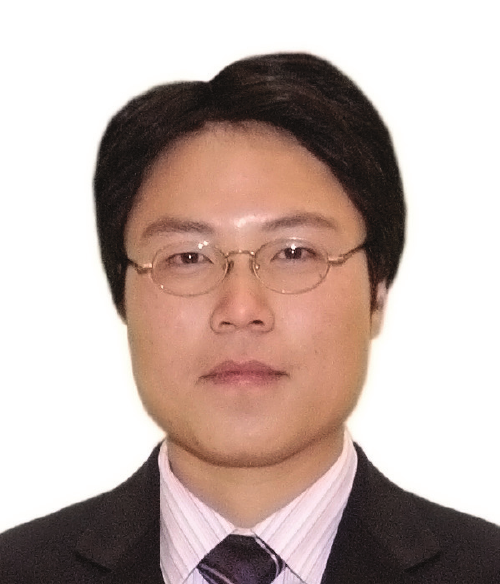}}] {Lin Gao} (S'08-M'10-SM'16) is an Associate Professor in the College of Electronic and Information Engineering at Harbin Institute of Technology, Shenzhen, China. He received M.S. and Ph.D. degrees in Electronic Engineering from Shanghai Jiao Tong University in 2006 and 2010, respectively. He was a Postdoc Research Fellow in the Network Communications and Economics Lab at The Chinese University of Hong Kong from 2010 to 2015. He received the IEEE ComSoc Asia-Pacific Outstanding Young Researcher Award in 2016. His research interests are in the interdisciplinary area combining telecommunications and microeconomics, with a particular focus on the game-theoretic and economic analysis for various communication and network scenarios, including cognitive radio networks, TV white space networks, cooperative communications, 5G communications, mobile crowd sensing, and Internet-of-Things.

\end{IEEEbiography}

\begin{IEEEbiography}[{\includegraphics[width=1in,height=1.25in,clip,keepaspectratio]{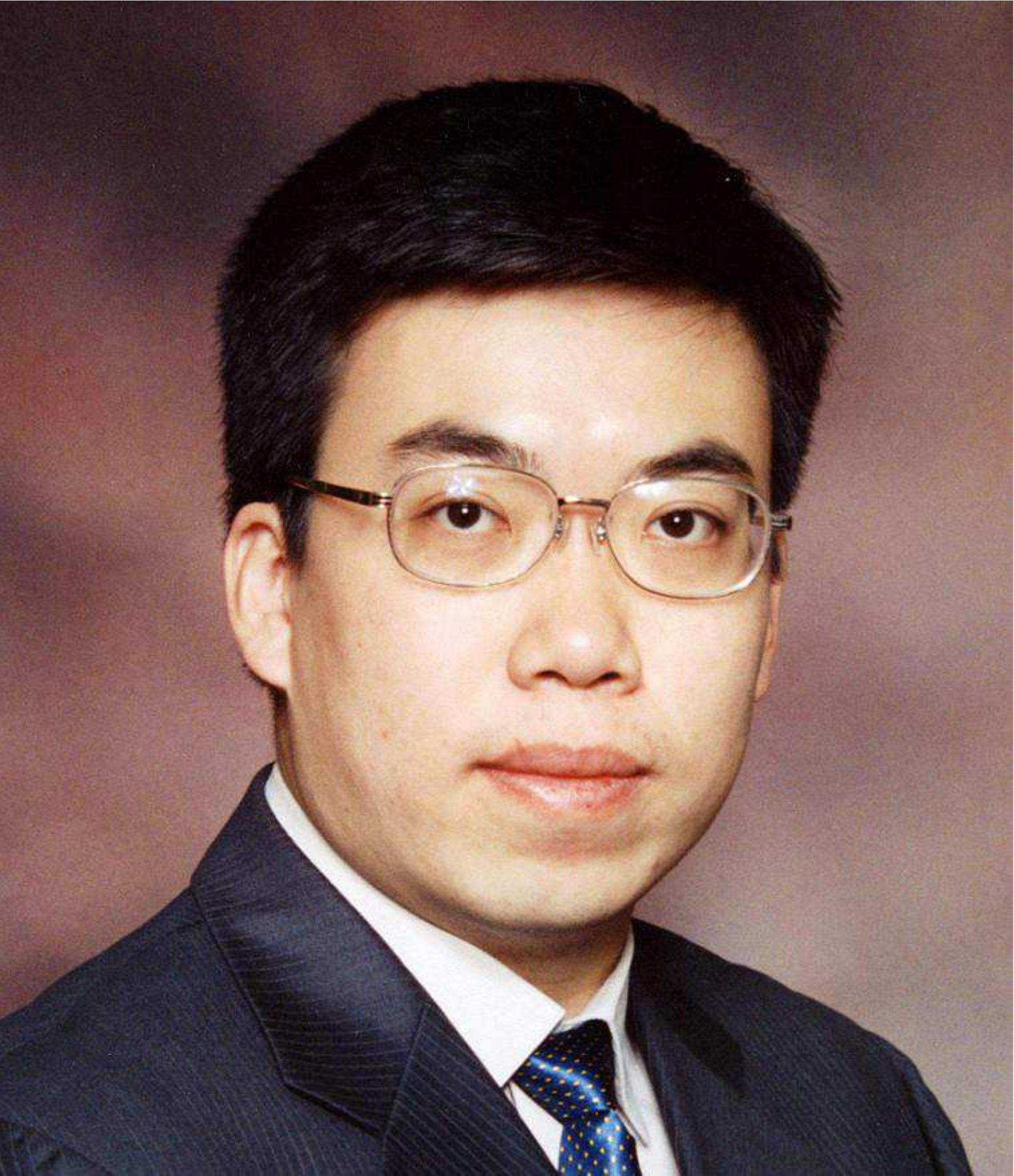}}]
	{Jianwei Huang} (S'01-M'06-SM'11-F'16) is an Associate Professor and Director of the Network Communications and Economics Lab (ncel.ie.cuhk.edu.hk), in the Department of Information Engineering at the Chinese University of Hong Kong. He received the Ph.D. degree from Northwestern University in 2005 and worked as a Postdoc Research Associate in Princeton during 2005-2007. He is the co-recipient of 8 international Best Paper Awards, including IEEE Marconi Prize Paper Award in Wireless Communications in 2011. He has co-authored six books, including the first textbook on "Wireless Network Pricing." He has served as an Associate Editor of IEEE Transactions on Cognitive Communications and Networking, IEEE Transactions on Wireless Communications, and IEEE Journal on Selected Areas in Communications - Cognitive Radio Series. He has served as the Chair of IEEE ComSoc Cognitive Network Technical Committee and IEEE ComSoc Multimedia Communications Technical Committee. He is a Fellow of IEEE, a Distinguished Lecturer of IEEE Communications Society, and a Thomson Reuters Highly Cited Researcher in Computer Science. 
\end{IEEEbiography}

\begin{IEEEbiography}[{\includegraphics[width=1in,height=1.25in,clip,keepaspectratio]{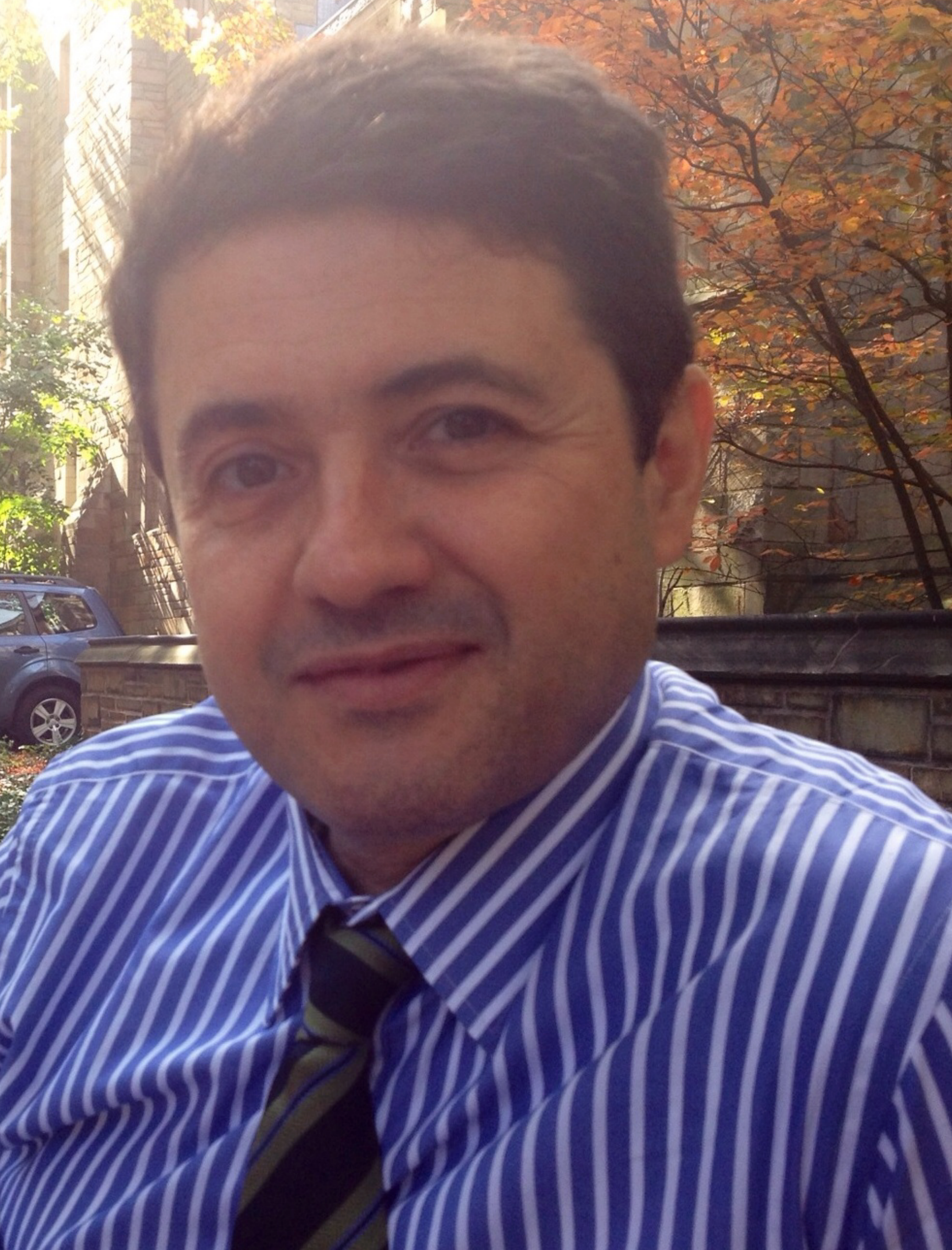}}] {Leandros Tassiulas} (S'89-M'91-SM'06-F'07) received the Ph.D. degree in electrical engineering 1085 from the University of Maryland, College Park, 1086 MD, USA, in 1991. He has held faculty positions 1087 with Polytechnic University, New York, the University of Maryland, College Park, and the University of Thessaly, Greece. He is currently the John C. Malone Professor of Electrical Engineering and member of the Institute for Network Science, Yale University, New Haven, CT, USA. His research interests are in the field of computer and communication networks with emphasis on fundamental mathematical models and algorithms of complex networks, architectures and protocols of wireless systems, sensor networks, novel Internet architectures, and experimental platforms for network research. His most notable contributions include the max-weight scheduling algorithm and the back-pressure network control policy, opportunistic scheduling in wireless, the maximum lifetime approach for wireless network energy management, and the consideration of joint access control and antenna transmission management in multiple antenna wireless systems. His research has been recognized by several awards, including the IEEE Koji Kobayashi Computer And Communications Award  (2016), the Inaugural INFOCOM 2007 Achievement Award for fundamental contributions to resource allocation in communication networks, the  INFOCOM 1994 best paper award, the National Science Foundation (NSF) Research Initiation Award (1992), the NSF CAREER Award (1995), the Office of Naval Research Young Investigator Award (1997), and the Bodossaki Foundation Award (1999).
\end{IEEEbiography}

\end{document}